\documentclass[a4paper,oneside,11pt]{amsart}
\usepackage{graphicx}
\usepackage[title]{appendix}
\usepackage{stmaryrd}
\usepackage{amsthm}
\usepackage{dsfont}
\usepackage{hyperref}
\usepackage{mathrsfs}
\usepackage{stmaryrd}
\usepackage[lite,initials,msc-links]{amsrefs}
\usepackage{amsmath}
\usepackage{centernot}
\usepackage{enumerate}
\usepackage{verbatim}
\usepackage{bm}
\usepackage{amsmath,amsfonts,amscd,amssymb}
\usepackage{longtable,geometry}
\usepackage{color,soul}

\usepackage{mathtools}

\newtheorem{theorem}{Theorem}[section]
\newtheorem{lemma}[theorem]{Lemma}
\newtheorem{proposition}[theorem]{Proposition}
\newtheorem{corollary}[theorem]{Corollary}

\numberwithin{equation}{section}
\theoremstyle{definition}

\newtheorem{example}{Example}
\theoremstyle{remark}
\newtheorem{remark}{Remark}

\newcommand{\IP}{\mathbf{P}}

\newcommand{\n}{\mathbf{n}}

\renewcommand{\Re}{\operatorname{Re}}
\renewcommand{\Im}{\operatorname{Im}}

\newcommand{\con}[1][]{\stackrel{#1}\leftrightarrow}

\newcommand{\ncon}[1][]{\overset{#1}{\not\leftrightarrow}}

\allowdisplaybreaks

\title{Circle patterns and critical Ising models}

\author{Marcin Lis}{
\address{Faculty of Mathematics\\ University of Vienna \\
Oskar-Morgenstern-Platz 1\\
1090 Wien}
\email{marcin.lis@univie.ac.at}}

\date{\today}

\keywords{Ising model, circle patterns, criticality}
\subjclass[2010]{82B20, 60C05, 05C50}

\begin{document}

\maketitle

\begin{abstract}
A circle pattern is an embedding of a planar graph in which each face is inscribed in a circle.
We define and prove magnetic criticality of a large family of Ising models on planar graphs whose dual is a circle pattern.
Our construction includes as a special case the critical isoradial Ising models of Baxter.

\end{abstract}
\section{Introduction}
The exact value of critical parameters is known only for a limited number of two-dimensional models of statistical mechanics.
One of the most prominent examples is the Ising model whose critical temperature on the square lattice
was famously computed by Kramers and Wannier~\cite{KraWan1} under a uniqueness hypothesis on the critical point. A confirmation of this condition came later as a corollary to the groundbreaking 
solution of the model provided by Onsager~\cite{Onsager}. 

The methods available at that time yielded also the critical temperature of inhomogeneous Ising models on the square lattice with different vertical and horizontal coupling constants.
A natural generalization of such models is the setting of arbitrary biperiodic graphs, i.e., weighted planar graphs whose group of symmetries includes~$\mathbb{Z}^2$. 
The critical point in the case of the square lattice with periodic coupling constants was first computed by Li~\cite{Li2012}, and the result was later
extended to all biperiodic graphs by Cimasoni and Duminil-Copin~\cite{CimDum}. Both approaches go through Fourier analysis of periodic matrices that arise from the combinatorial 
solutions of the Ising model due to Fisher~\cite{Fisher}, and Kac and Ward~\cite{KacWard} respectively. In particular, it is known that criticality in this setting is equivalent to the existence of nontrivial functions in the kernel of the associated Kac--Ward  matrix~\cite{Cimasoni2015}.

So far the only other class of planar graphs where critical parameters for the Ising model were explicitly known are the isoradial graphs defined by the condition 
that each face is inscribed in a circle with a \emph{common} radius. The critical Z-invariant coupling constants of Baxter~\cite{Baxter}
arise then as a solution to a system of equations requiring that the Ising model is invariant under a star-triangle transformation (which preserves isoradiality of the graph).
Criticality in the sense of statistical mechanics of the associated Ising model was proved in the periodic case in~\cite{CimDum}, and in the general case in~\cite{Lis2014a}.
Moreover, it is known that also in this setting the associated Kac--Ward matrix has a non-trivial kernel~\cite{Lis2014}. 

Isoradial graphs form the most general family of graphs where the critical Ising model was shown to be conformally invariant in the scaling limit~\cite{CheSmi12}.
The proof of Chelkak and Smirnov uses the fact that differential operators 
admit well behaved discretizations on isoradial graphs~\cite{Duffin, Mercat}. One should also mention that dimer models on isoradial graphs related to the Ising model were studied in~\cite{BdT1,BdT2,BdTR}.

A circle pattern is an embedding of a planar graph in which each face is inscribed in a circle of an \emph{arbitrary} radius. Circle patterns have been extensively studied 
in relation to discretely holomorphic functions~\cite{Bobenko1,Bobenko2,Bobenko3,Schramm}.
In this article we define coupling constants
for the Ising model on the dual graph of a circle pattern, and prove that the resulting model is critical.
Our construction, when restricted to isoradial graphs, recovers the critical Z-invariant coupling constants of Baxter.
Examples of new graphs where a critical Ising model can be defined in a local manner include, among many others, arbitrary trivalent graphs whose dual is a triangulation with acute angles.

Unlike the previous proofs of criticality~\cite{Li2012,CimDum,Lis2014a}, we do not invoke duality arguments. Instead,
we obtain exponential decay of correlations in the high-temperature regime using bounds on the operator norm of the Kac--Ward transition matrix~\cite{KLM, Lis2014a}.
To establish magnetic order in the low-temperature regime, we construct a non-trivial vector in the kernel of the critical Kac--Ward matrix (which readily implies infinite susceptibility), and then use known correlation inequalities to infer positive magnetization at lower temperatures.

Clearly, many of the natural questions about the model we study in this article remain open, 
and one of the most interesting is perhaps that of existence, conformal invariance and universality (among circle patterns) of the scaling limit.
The foundation for the proof of conformal invariance on isoradial graphs is the strong form of discrete holomorphicity that is satisfied by the fermionic observable 
associated with the critical Ising model. Here we provide the corresponding relations satisfied by the observable in the case of circle patterns.
These are only ``half'' of the relevant relations in the sense that the isoradial observable satisfies them on both the dual and primal graph, whereas the
generic observable does so only on the dual graph.

This article is organized as follows. In the next section we introduce the setup and state our main theorems. In Sect.~\ref{sec:KacWard}
we recall relevant results on the relation between the Ising model and the Kac--Ward matrix, and in Sect.~\ref{sec:proof} we provide the proofs of our results.
{\color{black} In Sect.~\ref{sec:holomorphic} we briefly discuss discrete holomorphicity of the critical fermionic observable on circle patterns, and in Sect.~\ref{sec:other} we present applications of our method to more general types of planar graphs (including s-embeddings of Chelkak~\cite{Chelkak})}.

\bigbreak\noindent\textbf{Acknowledgments} The author thanks Hugo Duminil-Copin for drawing his attention to, and explaining the inequality of Lemma~\ref{lem:DumCopTas}, and Dmitry Chelkak for
useful discussions about s-embeddings and possible extensions of the results contained in the first version of this article (see the discussion in Sect.~\ref{sec:other}).

This research was funded by EPSRC grants EP/I03372X/1 and EP/L018896/1 and was conducted when the author was at the University of Cambridge.

\section{Main results} \label{sec:main}
Let $\mathbb G=(V,E)$ and $\mathbb G^*=(V^*,E^*)$ be infinite, mutually dual, planar graphs embedded in the complex plane in such a way that each face of 
$\mathbb{G}^*$ is inscribed in a circle whose center is inside the closure of the face, and the vertices of $\mathbb{G}$ lie at the centers of the circles. 
We identify both $\mathbb G$ and $\mathbb G^*$ with their embedding and we say that $\mathbb{G}^*$ is a \emph{circle pattern}. (Note that we include the condition about the center being inside
the face in the definition of a circle pattern.)
For each $e=\{u,v\}\in E$, the dual edge $e^*$ is the common chord of the circles centered at $u$ and $v$. We denote by $\theta_{(u,v)}$ and $\theta_{(v,u)}$ half of the 
respective central angles given by the chord (see Fig.~\ref{fig:multiradial}). The ordered pairs $(u,v)$ and $(v,u)$ represent the two opposite directed versions of the undirected edge $\{u,v\}$. 
Note that for each $v\in V$, 
\begin{align} \label{eq:anglecond}
\sum_{u\sim v} \theta_{(v,u)} = \pi,
\end{align}
where the sum is over all vertices $u$ adjacent to $v$.
We will often assume the \emph{bounded angle property} of $\mathbb G^*$ by requiring that there exists $\varepsilon >0$ such that
\begin{align} \label{eq:regular}
\varepsilon \leq  \theta_{( u,v )} \leq \tfrac{\pi}2 -\varepsilon
\end{align}
for all directed edges $(u,v)$. 
This in particular implies that $\mathbb{G}$ is of bounded degree.
\begin{figure}
		\begin{center}
			\includegraphics[scale=0.85]{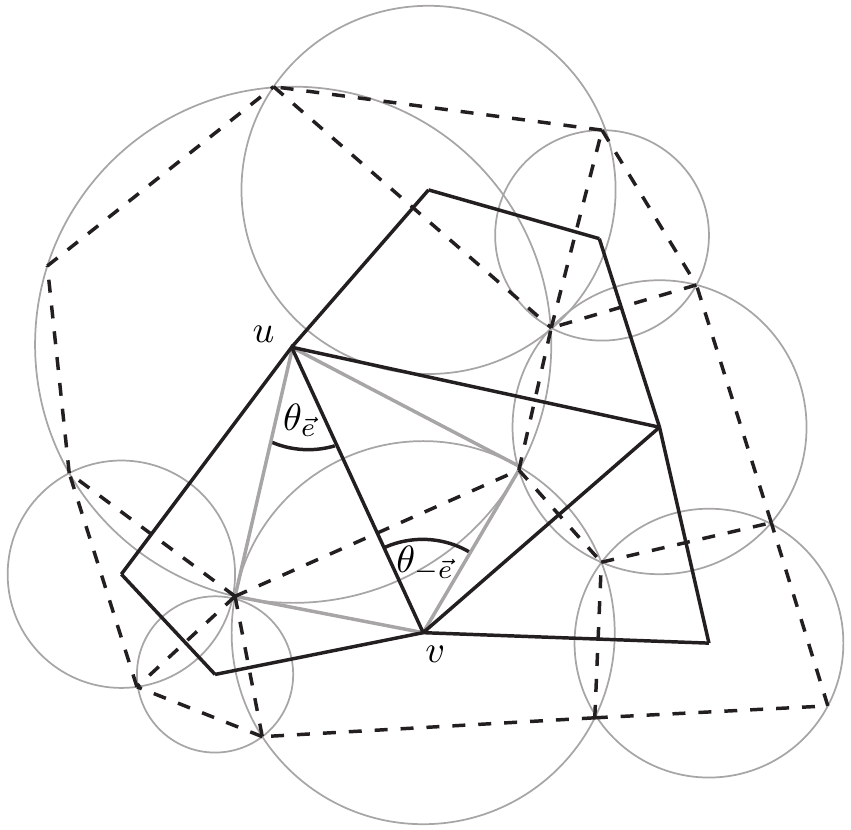}  
		\end{center}
		\caption{Local geometry of a circle pattern $\mathbb{G}^*$ and its dual $\mathbb{G}$ with dashed and solid edges respectively. Here, $\vec e=(u,v)$ and $-\vec e =(v,u)$}
	\label{fig:multiradial}
\end{figure}
Let $J=(J_e)_{e\in E}\in (0,\infty]^E$ be \emph{coupling constants} given by
\begin{align} \label{eq:cc}
\tanh J_{\{ u,v \}} = \sqrt{\tan \tfrac{\theta_{(u,v)}}2 \tan \tfrac{\theta_{(v,u)}}2}.
\end{align}
It is easy to see that the bounded angle property implies that $\|J \|_{\infty}<\infty$.

{\color{black}Note that if all the circumscribed circles have the same radius, then $\mathbb{G}$ and~$\mathbb{G}^*$ are isoradial, the angles satisfy $\theta_{(u,v)}=\theta_{(v,u)}$, and \eqref{eq:cc} defines the critical Z-invariant
coupling constants of Baxter~\cite{Baxter}.
On the other hand, the coupling constants \eqref{eq:cc} are a special case of the ones introduced by Chelkak \cite{Chelkak} in the setting of graphs called \emph{$s$-embeddings} which generalize circle patterns. 
However, no questions about criticality are asked in~\cite{Chelkak} and these are the main focus of the present article {\color{black}(In Sect.~\ref{sec:other} we discuss applications of our approach to general s-embeddings which yields partial answers to the question of criticality).}
Also, the model considered here is more general than the one of Bonzom, Costantino and Livine~\cite{BCL}, who studied a supersymmetric relation between the planar Ising model and spin networks.
The authors of \cite{BCL} considered the same coupling constants as~\eqref{eq:cc} but only in the case when $\mathbb G^*$ is a triangulation (and hence a circle pattern).}

We will study Ising models on finite connected subgraphs $G=(V_G,E_G)$ of $\mathbb{G}$ with coupling constants $J$.
To this end, let 
\[
\partial V_G = \{ v\in V_G: \exists u \in V \setminus V_G,\ \{u,v\} \in E\}
\]
be the \emph{boundary} of $G$, and let $\Omega_G=\{-1,+1\}^{V_G}$
be the space of \emph{spin configurations}. 
The \emph{Ising model}~\cite{Ising} at \emph{inverse temperature} $\beta>0$ with \emph{free} or \emph{`+' boundary conditions} conditions $\Box \in \{f,+\}$ is a probability measure on $\Omega_G$ given by
\begin{align*}
\IP^{\Box}_{G,\beta}(\sigma) = \frac{1}{Z^{\Box}_{G,\beta}}\exp\Big(\beta  \sum_{\{u,v\} \in E_G} J_{\{u,v\}}\sigma_u\sigma_v + 
 \mathbf{1}_{\{\Box=+\}}\beta \mathop{\sum_{v\in \partial V_G}  }_{u\notin V_G} J_{\{u,v\}}\sigma_v, 
 \Big),
 \end{align*}
where $Z^{\Box}_{G,\beta}$ is the normalizing constant called the \emph{partition function}.
We write $\langle\cdot  \rangle^{\Box}_{G,\beta}$ for the expectation with respect to $\IP^{\Box}_{G,\beta}$. By the second Griffiths inequality, we can define the \emph{infinite volume limits} of correlation functions by
\begin{align*}
\big \langle  \prod_{u\in A}  \sigma_u \big\rangle^{\Box}_{\mathbb{G},\beta}= \lim_{G\nearrow \mathbb{G}}\big \langle\prod_{u\in A}  \sigma_u \big \rangle^{\Box}_{G,\beta}, 
\end{align*}
where $A \subset V$ is finite, and
where the limit is taken over any increasing family of finite connected subgraphs $G$ containing $A$ and exhausting $\mathbb{G}$.

Let $d(\cdot,\cdot)$ be the graph distance on $\mathbb{G}$. The following is our main result identifying a phase transition at $\beta_c=1$: 
\begin{theorem} \label{thm:main}
Let $\mathbb{G}^*$ be a circle pattern satisfying the bounded angle property, and
consider the Ising model on $\mathbb{G}$ with coupling constants as in~\eqref{eq:cc}. Then
\begin{itemize}
\item[(i)] for every $\beta <1$, there exists $C_{\beta}>0$ such that for all $u,v\in V$,
\[
\langle \sigma_u \sigma_v \rangle^{f}_{\mathbb{G},\beta} \leq e^{-C_{\beta} d(u,v)},
\]
\item[(ii)] if the radii of circles are uniformly bounded from above, then for every $v\in V$, there exists $t_v>0$ such that for all $\beta > 1$, 
\[
\color{black}\langle \sigma_v \rangle^{+}_{\mathbb{G},\beta} \geq{1-\beta^{-t_v}}.
\]
Moreover, if the radii of circles are uniformly bounded from below, then one can choose one such $t=t_v$ for all $v$.
\end{itemize}
\end{theorem}
We also show that the magnetic susceptibility diverges at criticality:
\begin{theorem} \label{thm:susc}
Let $\mathbb{G}^*$ be as in Theorem~\ref{thm:main}, and assume that the radii of circles are uniformly bounded from below. Then for all $v\in V$,
\[
\chi_{\mathbb{G},\beta}(v)=\sum_{u\in V} \langle \sigma_u \sigma_v \rangle^{f}_{\mathbb{G},\beta}  < \infty 
\]
if and only if $\beta<1$.
\end{theorem}

Our main tool is the Kac--Ward transition matrix $\Lambda$ associated with the Ising model~\cite{KacWard}.
The exponential decay of correlations for $\beta <1$ will follow from our previous result~\cite{Lis2014a} which says that the Euclidean operator norm
of $\Lambda$ is strictly smaller than $1$. This part does not actually require the faces of $\mathbb{G}^*$ to be cyclic polygons and holds true 
for an arbitrary choice of angles $\theta_{(u,v)}$ satisfying condition \eqref{eq:anglecond}.
The complementary lower bound on the two-point function which identifies a phase transition at $\beta_c=1$ will follow from a construction of an eigenvector of $\Lambda$ with eigenvalue $1$ combined 
with known correlation inequalities.
This part crucially relies on the global geometry of the embedding.

We finish this section with a list of examples of circle patterns which are not necessarily isoradial. 

\begin{example} \label{ex:2}
Any \emph{acute triangulation} $\mathbb{G^*}$ of the plane, i.e., a planar graph whose all faces are acute triangles, is a circle pattern. In this case, $\mathbb{G}$ is a trivalent graph and
this is the set-up where the relation between the planar Ising model and spin networks was studied in \cite{BCL}. The authors heuristically derived the same coupling constants \eqref{eq:cc}
on the spin network side, and asked if they correspond to critical Ising models on a class of graphs that is larger than the isoradial graphs. 
Theorem \ref{thm:main} answers this question in the affirmative.
\end{example}

\begin{example} \label{ex:3} 
A \emph{circle packing} is a representation of a planar graph $\mathbb G$ where the vertices are the centers of interior-disjoint disks in the plane, and two vertices
are adjacent if the respective discs are tangent.
Consider a circle packing of an infinite planar graph $\mathbb G$ such that each face is {convex}. Then the dual graph $\mathbb{G}^*$ can be simultaneously circle-packed 
in such a way that the circles centered at the endpoints of an edge $e$ and its dual edge $e^*$ meet orthogonally at one point (see Fig.~\ref{fig:packing}). 
Let $Q$ be the {quadrangulation} whose vertices are the vertices and faces of $\mathbb G$, and whose edges are of the form $\{u,u^*\}$, 
where $u$ is a vertex of $\mathbb G$ and $u^*$ is one of the faces incident on $u$.
Then $Q^*$ is the \emph{medial graph} of $\mathbb G$ with vertices
given by the meeting points of the circles, and with edges connecting every pair of consecutive vertices around every circle. By construction, $Q^*$ is a circle pattern.
\begin{figure} 
		\begin{center}
			\includegraphics[scale=0.85]{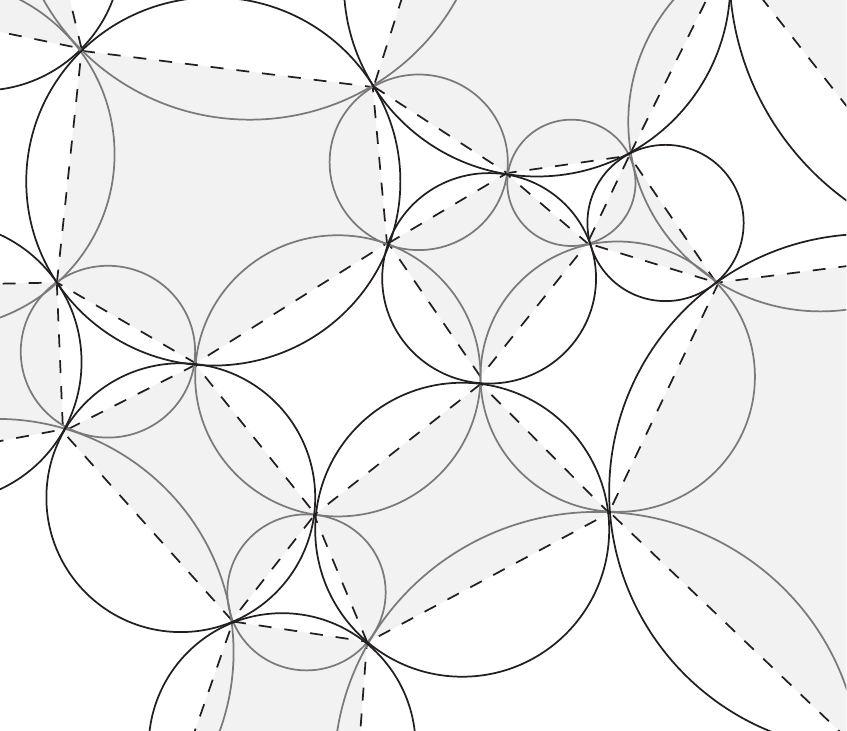}  
		\end{center}
		\caption{A piece of a circle packing and its dual. The quadrangulation $Q$ from Example~\ref{ex:3} is the graph whose vertices are the regions bounded by dashed lines.
		Its dual $Q^*$ is a circle pattern}
	\label{fig:packing}
\end{figure}

\end{example}

\begin{example} \label{ex:4} 
Let $\mathbb{G}$ and $\mathbb{G}^*$ be be mutually planar graphs embedded in such a way that each pair of a primal and dual edge meets at a right angle. 
Let $Q$ be the quadrangulation whose vertices are the vertices, faces and edges of $\mathbb{G}$, and whose edges are of the form $\{e,v\}$ where $e$ is an edge of $\mathbb{G}$ and $v$ is either a vertex or a face incident on $e$.
Note that each face of $Q$ has at least two right interior angles, and hence $Q$ is a circle pattern. 
The dual graph $Q^*$ is the $4$-regular graph that is obtained from $\mathbb{G}$ by replacing each edge of $\mathbb{G}$ by a quadrilateral and each vertex of degree $d$ by a $d$-gon face.
\end{example}

\begin{example} \label{ex:5}
Let $\mathbb{G} =\mathbb{Z}^2$ and fix $0<a\leq b<\infty$. For the $i$-th column of vertical edges, choose a coupling constant $J_i \in [a,b]$.
To each horizontal edge between column $i$ and $i+1$ assign a coupling constant $J_{i,i+1}$ satisfying 
\[
\tanh J_{i,i+1} = e^{-J_i-J_{i+1}}.
\]
The resulting Ising model is critical as these coupling constants are computed using \eqref{eq:cc} applied to the dual lattice $\mathbb{G}^*\simeq \mathbb{Z}^2$ with appropriately stretched or squeezed columns.
Note that if $J_i $ are constant, the classical anisotropic critical Ising model is recovered, and 
upon setting $J_i=-\frac12 \log(\sqrt{2}-1)$ we get the homogenous critical model.  
\end{example}

\section{The Kac--Ward operator and the fermionic observable} \label{sec:KacWard}
{\color{black}In this section we define the Kac--Ward matrix on a general graph in the plane and relate its inverse to the spin fermionic observable of Smirnov~\cite{Lis2014,Cimasoni2015}.
We also recall the exact value of the Euclidean operator norm of the associated transition matrix computed in~\cite{Lis2014a} (Lemma~\ref{lem:normbound}). Using this we later prove in Lemma~\ref{lem:normbound1} that in the 
setting of Theorem~\ref{thm:main} the operator norm is smaller than~$1$ if and only if $\beta<1$.
The novel contribution of this section is that at $\beta=1$ we construct an eigenvector of eigenvalue $1$ of the transition matrix, or equivalently, a null vector of the Kac--Ward matrix (Proposition~\ref{thm:eigenvector}).
{\color{black} The consequences of the existence of such a vector for biperiodic graphs and s-embeddings are discussed in Sect.~\ref{sec:other}}.

Let $G=(V_G,E_G)$ be a graph drawn in the plane with possible edge crossings. Let $\vec E_G$ be the set of all directed versions of the edges in $E_G$, and let $\vec x=(\vec x_{\vec e})_{\vec e \in \vec E_G}$ be a vector of associated real weights.
For a directed edge $\vec e=(u,v) $, we write $t_{\vec e}=u$ for its {tail}, $h_{\vec e}=v$ its {head}, $-\vec e=(v,u)$ its {reversal},
and $e =\{u,v\} \in E_G$ for its undirected equivalent.}

The \emph{Kac--Ward transition} matrix $\Lambda_G(\vec x)$ is 
a matrix indexed by $\vec E_G$ given by
\begin{align} \label{eq:transition}
\Lambda_G(\vec x)_{\vec e, \vec g} = \begin{cases}
		\vec x_{-\vec e} \vec x_{\vec g} e^{\frac{i}{2}\angle(\vec e, \vec g)}
		& \text{if $h_{\vec e} = t_{\vec g}$ and $\vec g \neq -\vec e$}; \\
		0 & \text{otherwise},
	\end{cases}
\end{align}
where 
\begin{align*}
\angle(\vec e, \vec g)= \text{Arg}\Big(\frac{h_g-t_g}{h_e-t_e}\Big)  \in  (-\pi,\pi]
\end{align*} 
is the \emph{turning angle} from $\vec e$ to $\vec g$. Here we identified the vertices with the corresponding complex numbers given by the embedding.
The \emph{Kac--Ward matrix} is defined by 
\[
T_G(\vec x)=\textnormal{Id}-\Lambda_G(\vec x),
\] 
where $\textnormal{Id}$ is the identity matrix indexed by $\vec E_G$.

We now assume that $G$ is finite.
Let $\mathcal{E}_G$ be the collection of \emph{even subgraphs} of $G$, i.e., subsets $\omega$ of $E_G$ such that all vertices in $V_G$ have even degree in (the subgraph induced by) $\omega$.
Let $ x=( x_{ e})_{ e \in  E_G}$ be weights on the undirected edges given by 
\begin{align} \label{eq:factorization}
x_e= \vec x_{\vec e} \vec x_{-\vec e},
\end{align}
where $\vec e, -\vec e$ are the two directed versions of $e$. Consider the partition function
\begin{align} \label{eq:KW}
Z_G(x)= \sum_{\omega \in\mathcal{E}_G} (-1)^{C({\omega})} \prod_{e\in \omega}x_e,
\end{align}
where $C(\omega)$ is the number of pairs of edges in $\omega$ that cross. {\color{black}The seminal identity of Kac and Ward~\cite{KacWard} (which in the general form allowing edge crossings was established in~\cite{Sherman1, DolEtAl, KLM}) reads}
\begin{align} \label{eq:KacWard}
Z_G(x) = \sqrt{{\det} T_G(\vec x)}.
\end{align}
If $G$ has no edge crossings, then the high-temperature expansion of the Ising partition function $Z^{f}_{G,\beta}$ is equal, up to an explicit constant,
to $Z_G(x)$ with $x_e =\tanh \beta J_e$, and hence \eqref{eq:KacWard} establishes an intrinsic relation between the Ising model and the Kac--Ward operator.

It turns out that one can go further and also express the inverse of $T_G(\vec x)$ in terms of related partition functions which, in addition to an even subgraph, involve a path weighted by a complex factor.
We now define these objects which are forms of the \emph{fermionic observable} introduced by Smirnov~\cite{smirnov}.
For a directed edge $ \vec e \in \vec E$, let $m_{ \vec e}= (t_{\vec e}+h_{\vec e})/2$ be its {midpoint}.
For $\vec e,\vec g \in \vec E_G$, $\vec e\neq -\vec g$, we define a modified graph $G_{\vec e, \vec g}$ with vertex set $V_G \cup \{ m_{\vec e}, m_{ \vec g}\}$
and edge set $(E_G \setminus \{ e,g\}) \cup \{ \{m_{\vec e},h_{\vec e}\},\{t_{\vec g},m_{\vec g}\}\}$, and define the weights of the undirected \emph{half-edges} to be 
$x_{\{m_{\vec e},h_{\vec e}\}} = \vec x_{\vec e}$ and $x_{\{t_{\vec g},m_{\vec g}\}} = \vec x_{-\vec g}$.
Let $\mathcal{E}_G(\vec e, \vec g)$ be the collection of sets of edges $\omega$ of $G_{\vec e, \vec g}$ containing $ \{m_{\vec e},h_{\vec e}\}$ and $\{t_{\vec g},m_{\vec g}\}$, 
and such that all vertices in $V_G$ have even degree in $\omega$. (Note that we do not require that $m_{\vec e}$ and $m_{\vec g}$ have even degree.)
From standard parity arguments, it follows that each $\omega \in \mathcal{E}_G(\vec e, \vec g)$ contains a self-avoiding path starting at $m_{\vec e}$ and ending at $m_{\vec g}$.
We denote by $\gamma_{\omega}$ the left-most such path, i.e., {\color{black}a path which at every step takes the left-most possible turn}.
We also write $\alpha(\gamma_{\omega})$ for the total turning angle of $\gamma_{\omega}$, i.e., the sum of turning angles between consecutive steps of $\gamma_{\omega}$.

Let $F_G(\vec x)$ be a matrix indexed by $\vec E_G$ given by
\begin{align} \label{eq:fo}
F_{G}(\vec x)_{\vec e, \vec g} = 
 \begin{cases}
\text{Id}_{\vec e, \vec g} + \frac{1}{Z_G(x)} \sum_{\omega \in  \mathcal{E}_G(\vec e, \vec g) }  e^{-\tfrac{i}2 \alpha(\gamma_{\omega})} \prod_{e \in \omega} x_e, & \text{if } \vec g \neq -\vec e; \\
 0 & \text{otherwise,}
\end{cases}
\end{align}
where $\vec x$ and $x$ are related via~\eqref{eq:factorization}. For a graph $G$ with no edge crossings, the following identity:
\begin{align} \label{eq:lis}
\overline{F_G(\vec x)}=T_G^{-1}(\vec x)
\end{align}
which gives foundations for our results was proved in~\cite{Lis2014,Cimasoni2015}, and is valid for any weight vector $x$.
{\color{black}We refer the interested reader to \cite{CCK} for a recent detailed account of the relationship of \eqref{eq:lis} with other combinatorial approaches to the Ising model.}

The main new contribution of this section is  the following construction of a nontrivial vector in the kernel of the critical Kac--Ward matrix defined on the dual of a circle pattern. 

\begin{proposition} \label{thm:eigenvector}
Let $\mathbb G^*$ be a circle pattern, and let $\vec x$ be weights on the directed edges of $\mathbb G$ given by $\vec x_{\vec e}=\sqrt{\tan\tfrac{ \theta_{\vec e}}2}$, $\vec e\in \vec E$.
For an edge $e^*$ of $\mathbb G^*$, let $|e^*|$ be its length, and define $\rho_{\vec e} =\rho_{-\vec e}  =\sqrt{|e^*|}$,
where $\vec e$ and $-\vec e$ are the two directed edges crossing $e^*$.
Then
\[
\Lambda_{\mathbb{G}}(\vec x) \rho = \rho.
\]
\end{proposition}
This identity will directly follow from the next lemma. Let $\mathbb G=(V,E)$, and let $\vec E$ be the set of directed edges of $\mathbb G$.
We define $\textnormal{In}_v = \{ \vec e \in \vec E: h_{\vec e} =v \}$ and $\textnormal{Out}_v = \{ \vec e \in \vec E: t_{\vec e} =v \}$ to be the directed edges pointing at and away from the vertex $v$ respectively.
Let $J$ be the involutive automorphism of $\mathbb{C}^{\vec E}$ induced by the map $\vec e \mapsto -\vec e$, and let $\tilde \Lambda(\vec x) = J \Lambda_{\mathbb{G}}(\vec x) $. As was noted in~\cite{Lis2014a}, and is easily seen from 
the definition of the transition matrix~\eqref{eq:transition}, $\tilde \Lambda(\vec x) $ is
a Hermitian, block-diagonal matrix with blocks $\tilde \Lambda(\vec x) ^v$, $v\in V$, acting on the linear subspace indexed by $\textnormal{Out}_v$, and given by
\begin{align*}
\tilde \Lambda(\vec x) ^v_{\vec e, \vec g} = \begin{cases}
		\vec x_{\vec e} \vec x_{\vec g} e^{\frac{i}{2}\angle(-\vec e, \vec g)}
		& \text{if $t_{\vec e} = t_{\vec g}=v$ and $\vec e\neq \vec g$}; \\
		0 & \text{otherwise}.
	\end{cases}
\end{align*}
Let $\rho^v$ be the restriction of $\rho$ to the subspace indexed by $\text{Out}_v$. Note that $J \rho=\rho$, and hence to prove Proposition~\ref{thm:eigenvector}, it is enough to show the following:
\begin{lemma} Let $\vec x$ be as in Proposition~\ref{thm:eigenvector}. Then for all $v \in V$, 
\[
\tilde \Lambda(\vec x)^v  \rho^v = \rho^v.
\]
\end{lemma}
\begin{proof}
{\color{black}Fix $\vec e \in \text{Out}_v$, and let $\vec e, \vec e_1, \ldots, \vec e_n $ be a clockwise ordering of $\text{Out}_v $.
Then
\begin{align*}
\angle( -\vec e, \vec e_k) &=\pi- \theta_{\vec e}-2\sum_{j=1}^{k-1} \theta_{\vec e_j}- \theta_{\vec e_k}, \qquad \text{for } k=1,\ldots,n.
\end{align*}
Without loss of generality we can assume that the radius of the circle centered at $v$ is equal to $1$, and hence $\rho_{\vec g} = \sqrt{2\sin \theta_{\vec g}}$ for all $\vec g\in \text{Out}_v$. We have
\begin{align*}
\frac{[\tilde \Lambda(\vec x)^v \rho^v]_{\vec e}}{\vec x_{\vec e}} &= \sum_{\vec g \in \text{Out}_v \setminus \vec e} \rho_{\vec g} \vec x_{\vec g} e^{\frac{i}{2}\angle(-\vec e, \vec g)} \\
&=\sqrt2\sum_{\vec g \in \text{Out}_v \setminus \vec e} \sqrt{\sin  \theta_{\vec g}\tan\tfrac{\theta_{\vec g}}2}e^{\frac{i}{2}\angle(-\vec e, \vec g)}  \\
&=2 \sum_{\vec g \in \text{Out}_v \setminus \vec e} \sin  \tfrac{\theta_{\vec g}}2e^{\frac{i}{2}\angle(-\vec e, \vec g)}  \\
&= -i \sum_{k=1}^n \big(e^{\frac{i}2\theta_{\vec e_k}}-e^{ -\frac{i}2\theta_{\vec e_k}}\big)e^{\frac{i}2\big(\pi- \theta_{\vec e}-2\sum_{j=1}^{k-1} \theta_{\vec e_j}- \theta_{\vec e_k}\big)} \\&
= -i (e^{\frac{i}2(\pi- \theta_{\vec e})}-e^{-\frac{i}2(\pi- \theta_{\vec e})})\\
&= 2\cos \tfrac {\theta_{\vec e}}2,
\end{align*} 
where the second last equality follows from a telescopic sum.}
Hence, 
\[[\tilde \Lambda(\vec x)^v \rho^v]_{\vec e} = 2 \cos \tfrac {\theta_{\vec e}}2 \vec x_{\vec e}=2\cos \tfrac {\theta_{\vec e}}2 \sqrt{\tan \tfrac{\theta_{\vec e}}2}=\sqrt{2\sin \theta_{\vec e}} =\rho^v_{\vec e}. \qedhere
\]
\end{proof}

We finish this section with an upper bound on the operator norm of the transition matrix which we will use in the next section to prove exponential decay of the two-point function in the high-temperature ($\beta<1$) regime.
\begin{lemma}  \label{lem:normbound} Let $(\theta_{\vec e})_{\vec e\in \vec E}$ satisfy condition~\eqref{eq:anglecond}, i.e., $\sum_{u\sim v} \theta_{(v,u)} = \pi$ for all $v\in V$, and let $\vec x$ be weights on the directed edges of $\mathbb G$ such that
\[|\vec x_{\vec e}| \leq \sqrt{ s \tan\tfrac{ \theta_{\vec e}}2}\] for some $s>0$ and all $\vec e \in \vec E$.
Then the induced operator norm of $\Lambda_{\mathbb G}(\vec x)$ acting on $\ell^2(\vec E)$ satisfies 
\begin{align*} 
\| \Lambda_{\mathbb G}(\vec x) \| \leq s.
\end{align*}
\end{lemma}
\begin{proof}
Note that $\| \tilde \Lambda(\vec x) \|= \|\Lambda_{\mathbb G}(\vec x)\|$. Since $ \tilde \Lambda(\vec x) $ is block diagonal with blocks $\tilde \Lambda(\vec x)^v$, the desired inequality
is a direct consequence of condition \eqref{eq:anglecond} and Lemma~2.4 of~\cite{Lis2014a} which says that
the operator norm of $\tilde \Lambda(\vec x)^v$ acting on the $\deg(v)$-dimensional Euclidean complex vector space indexed by $\text{Out}_v$ 
is the positive solution $s$ of
\begin{equation} \label{eq:lisnorm}
\sum_{\vec e \in \textnormal{Out}_v } \arctan\Big(\frac {\vec x_{\vec e}^2}{s} \Big) =\frac{\pi}2. \qedhere
\end{equation}
\end{proof}

\section{Proofs of main results} \label{sec:proof}
{\color{black}
In this section we prove the main theorems. In preparation for the proofs,
we bound the Ising two-point correlation functions with free boundary conditions from both below and above by the entries of the inverse Kac--Ward matrix (Lemma~\ref{lem:fermionbound}). The upper bound together with the previous
bounds on the operator norm of the transition matrix will be used to obtain exponential decay of correlations for $\beta<1$.
We then recall a general differential inequality due to Duminil-Copin and Tassion (Lemma~\ref{lem:DumCopTas}) which gives a lower bound for the magnetization in terms of a correlation function $\varphi$ defined in~\eqref{eq:phi}. We use the existence of the eigenvector 
from Proposition~\ref{thm:eigenvector} together with the lower bound from Lemma~\ref{lem:fermionbound} to obtain a non-zero lower bound on $\varphi$ at $\beta=1$ (Lemma~\ref{lem:phi}).
Finally, we can integrate the differential inequality to obtain positive magnetization for all $\beta>1$.

We first need to state the classical \emph{high-temperature expansion} of the Ising correlation functions. To this end, for $u,v\in V_G$, let $\mathcal{E}_{G}(u,v)$ be the collection of subsets $\omega$ of $E_G$ such that each vertex in $V_G \setminus(\{u\} \Delta \{v\})$ (resp.\ in $\{u\} \Delta \{v\}$) 
has even (resp.\ odd) degree in $\omega$, where $\Delta$ is the symmetric difference. 
We define the partition function
\[
Z_G(x)_{u,v} = \sum_{\omega \in \mathcal{E}_{G}(u,v)} \prod_{e\in \omega} x_e.
\]
As observed by van der Waerden~\cite{vdW}, for $x_{e}= \tanh(\beta J_e)$, and all $u,v \in V_G$, we have
\begin{align} \label{eq:hte}
\langle \sigma_u \sigma_v \rangle_{G,\beta}^{f} = \frac{Z_{G}( x)_{u,v}}{Z_{G}( x)}.
\end{align}}

{\color{black}Using the relationship between the fermionic observable and the inverse Kac--Ward matrix~\eqref{eq:lis}, we can now prove two-sided bounds on the two-point correlation functions 
in terms of the entries of inverse Kac--Ward matrices with possible signed weights. To this end, given a set of signs $ \tau \in \{- 1, 1 \}^{\vec E}$ on the directed edges, we define the signed weight vector $\vec x^{\tau}$ on $\vec E$ simply by
$\vec x^{\tau}_{\vec e}= \tau_{\vec e} \vec x_{\vec e}$.
\begin{lemma} \label{lem:fermionbound} Let $\beta>0$, $J_e> 0$ and $x_{e}= \tanh(\beta J_e)$, $e\in E$. Then for all $u,v \in V_G$, $u\neq v$,
there exists $ \tau \in \{- 1, 1 \}^{\vec E}$ such that for all weight vectors $\vec x > 0$ on $\vec E$ related to $x$ by \eqref{eq:factorization}, we have
\[
 \mathop{\max_{\vec e \in \textnormal{In}_u,  \vec g \in \textnormal{Out}_v }}_{\vec e \neq \vec g}  
  \frac{ | T^{-1}_G(\vec x)_{\vec e, \vec g}|}{\vec x_{\vec e}\vec x_{-\vec g}} \leq 
\langle \sigma_u \sigma_v \rangle_{G,\beta}^f \leq \sum_{\vec e \in \textnormal{In}_u, \vec g \in \textnormal{Out}_v } \vec x_{-\vec e}\vec x_{\vec g}
{ | T^{-1}_G( \vec x^{\tau})_{\vec e, \vec g}|}.
\]
\end{lemma}}

{\color{black}Note that the lower bound is given by the inverse Kac--Ward matrix with no additional signs, and the upper bound requires signed weights.}
\begin{proof}
We first prove the lower bound.
By~\eqref{eq:lis}, we have that $T_G^{-1}(\vec x)=\overline{F_G(\vec x)}$. By the definition of $\mathcal{E}_G(\vec e, \vec g)$, removing the two half-edges from $\omega\in \mathcal{E}_G(\vec e, \vec g)$ (carrying weights $\vec x_{\vec e}$ and $\vec x_{-\vec g}$)
yields a configuration in $\mathcal{E}_{G}(u,v)$. It is now enough to use the high temperature expansion~\eqref{eq:hte} and the fact that all terms in the definition of $Z_G(x)_{u,v} $ are positive, whereas 
the corresponding terms in $\overline{F_G(\vec x)}_{\vec e, \vec g}$ carry a complex factor.

To obtain the upper bound, we construct an augmented graph $G^{\gamma}$ by adding to $G$ a simple path $\gamma$ connecting $u$ and $v$
in such a way that $\gamma$ crosses each edge at most once and does not pass through any vertex of $G$.
Without loss of generality, we can assume that $\gamma$ is the single edge $\{u,v\}$.
We fix $\tau \in  \{- 1,1\}^{\vec E}$ satisfying
\[
\tau_{\vec e} \tau_{-\vec e} = \begin{cases}
-1 & \text{ if $e$ is crossed by $\gamma$,} \\
+1 & \text{ otherwise,}
\end{cases}
\]
and, when needed, we extend the weights to $G^{\gamma}$ by setting
$\tau_{\vec \gamma} =\tau_{-\vec \gamma}=1$. We also chose an orientation $\vec \gamma = (u,v)$ of $\gamma$.

{\color{black}We claim that for all $\vec x^{\tau}$ related to $x^{\tau}$ by \eqref{eq:factorization}, we have
\begin{align} \label{eq:fb1}
\langle \sigma_u \sigma_v \rangle_{G,\beta}^f  &= \frac{1}{Z_G(x)} \frac{\partial }{\partial \vec x_{\vec \gamma} } Z_{G^{\gamma}} ( x^{\tau}) \Big |_{\vec x_{\vec \gamma}=0,\vec x_{-\vec \gamma}=1}.
\end{align} 
To understand this equality first note that the differentiation selects only the even subgraphs of $G^{\gamma}$ that actually contain $\gamma$. 
Furthermore, there is a bijection between such subgraphs and the graphs in $\mathcal{E}_{G}(u,v)$ (it is enough to remove~$\gamma$).
Hence, \eqref{eq:fb1} follows from the high-temperature expansion~\eqref{eq:hte} and the fact that the signs of $x^{\tau}$ are chosen in such a way that they cancel out the signs appearing in the definition \eqref{eq:KW} of $Z_{G^{\gamma}} ( x^{\tau})$.

Using the Kac--Ward formula~\eqref{eq:KacWard} for $G^{\gamma}$ and Jacobi's formula for the derivative of a determinant, we get that the right-hand side of \eqref{eq:fb1} is equal to
\begin{align*}
\frac{1}{Z_G(x)} &\frac{\partial }{\partial \vec x_{\vec \gamma} } \sqrt{{\det} T_G(\vec x)} \Big |_{\vec x_{\vec \gamma}=0,\vec x_{-\vec \gamma}=1} \\ &=
\frac{1}{Z_G(x)} \frac 12 \sqrt{{\det}T_{G^{\gamma}}(\vec x^{\tau})} \text{Tr} \Big (T^{-1}_{G^{\gamma}}(\vec x^{\tau})  \frac{\partial }{\partial \vec x_{\vec \gamma}}  T_{G^{\gamma}}(\vec x^{\tau})\Big)\Big |_{\vec x_{\vec \gamma}=0,\vec x_{-\vec \gamma}=1} \\
&=\frac{Z_{G}(x^{\tau})}{Z_G(x)}  \frac 12\text{Tr} \Big (T^{-1}_{G^{\gamma}}(\vec x^{\tau})  \frac{\partial }{\partial \vec x_{\vec \gamma} }  T_{G^{\gamma}}(\vec x^{\tau})\Big)\Big |_{\vec x_{\vec \gamma}=0,\vec x_{-\vec \gamma}=1} \\
 &=\frac{Z_{G}(x^{\tau})}{Z_G(x)} \sum_{\vec e \in \textnormal{Out}_v, \vec g \in \textnormal{In}_u }
 e^{\tfrac{i}2 (\angle(\vec\gamma, \vec e)+\angle(\vec g,\vec \gamma))} \vec x_{-\vec e}\vec x_{\vec g}
{  T^{-1}_G( \vec x^{\tau})_{\vec e, \vec g}} \\
 &\leq \sum_{\vec e \in \textnormal{Out}_v , \vec g \in \textnormal{In}_u}   \vec x_{-\vec e}\vec x_{\vec g}
{ | T^{-1}_G( \vec x^{\tau})_{\vec e, \vec g}|}.
\end{align*}

The second equality follows from the Kac--Ward formula again and the fact that $x_{\gamma}=\vec x_{\vec \gamma}\vec x_{-\vec \gamma}=0$. The third equality follows from the definition of $T_G( \vec x^{\tau})$, and the last inequality holds true since 
$Z_{G^{\gamma}}(x^{\tau})$ counts even subgraphs with additional signs, whereas all terms in $Z_G(x)$ are positive.}
\end{proof}

{\color{black}
The second crucial ingredient in the proof of our Theorem~\ref{thm:main} will be a differential inequality for the magnetization of general Ising models due to Duminil-Copin and Tassion. To state it, we need to introduce some notation. To this end,
let $S\subset V$ be such that the subgraph of $\mathbb{G}$
induced by $S$ is connected. 
We will denote by $\langle \cdot  \rangle^{\Box}_{S,\beta}$ the expectation with respect to the Ising model on this induced subgraph.
For $v\in V$, we define
\begin{align} \label{eq:phi}
\varphi_{S,\beta} (v) = \sum_{u\in S} \sum_{w \notin S}\tanh \big( \beta J_{\{u,w\}}\big) \langle \sigma_v \sigma_u  \rangle^f_{S,\beta}.
\end{align}
The next lemma is contained in an unpublished version~\cite{DumCopTasU} of~\cite{DumCopTas}, and we give its proof for completeness in Appendix~\ref{app:a}. }
\begin{lemma} \label{lem:DumCopTas}
For any finite subgraph $G=(V_G,E_G)$, any $v\in V_G$, and $\beta >0$, 
\[
\frac{d}{d \beta} \langle \sigma_v \rangle_{G,\beta}^+  \geq \tfrac 1 {\beta} \mathop{\inf_{S \ni v}}_{S\subset V_G} \varphi_{S,\beta}(v)\big(1- ( \langle \sigma_v \rangle_{G,\beta}^+ )^2  \big).
\]
\end{lemma}

{\color{black}
We will next use Lemma~\ref{lem:fermionbound} together with the existence of the eigenvector from Proposition~\ref{thm:eigenvector} to get a uniform lower bound on $\varphi_{S,\beta}(v)$.
This may be considered as the main technical novelty of the present article, and its consequences for the existence of positive magnetization at low temperatures for other families of planar graph are discussed in Sec.~\ref{sec:other}.
}

Before stating the result, we need to introduce additional notation. 
Let $G=(V_{G},E_{G})$ be a finite subgraph of $\mathbb{G}$ induced by the vertex set $V_{G}$. Consider the edges of $\mathbb{G}$ whose
one endpoint is in $V_{G}$ and the other one in $V \setminus V_{ G}$. Each such edge splits into two half-edges, one of which is incident to $V_{G}$.
We add the incident half-edges to the edge set $E_{G}$, and we assign them the weights of the corresponding full edges in $G$. We also add their endpoints (which are the midpoints 
of the corresponding edges of $\mathbb{G}$) to the vertex set~$V_{G}$ and we call the resulting graph $\bar{G}=({V}_{\bar G},{E}_{\bar G})$.
Note that by construction, all vertices in~$V_{G}$ are \emph{interior} in $\bar{G}$ meaning that they have the same degree in $\bar{G}$ as in $\mathbb{G}$.
Let $\vec{n}(\partial G)$ be the set of the directed versions of the half-edges in $\bar{E}_{G}$
which point outside~$G$. 

\begin{lemma} \label{lem:phi}
Let $\mathbb G^*$ be a circle pattern satisfying the bounded angle property~\eqref{eq:regular}, and such that the radii of all circles are uniformly bounded from above by $R<\infty$. 
Then for every $v\in V$ and every finite $S\subset V$ containing $v$,
\[
\varphi_{S,1} (v)  \geq \sqrt{\frac{ r }{R}} \tan \tfrac{\varepsilon}2,
\]
where $r$ is the maximal radius of a circle centered at a neihbor of $v$, and $\varepsilon$ is as in \eqref{eq:regular}.
\end{lemma}
\begin{proof}
For $d(v,V\setminus S)= 1$, the bound easily follows from the definition of $\varphi_{S,1} (v)$ and the bounded angle property.
We can therefore assume that $d(v,V\setminus S)> 1$.

Let $G$ be the subgraph of $\mathbb G$ induced by the vertices in $S$. Let $\vec x_{\vec e} =\sqrt{\tan \tfrac{\theta_{\vec e}}2} $ and let $\rho$ be the eigenvector of $\Lambda_{\mathbb G}(\vec x)$ as defined in Proposition~\ref{thm:eigenvector}
restricted to the directed edges of $\bar G$ (here the half-edges are identified with their counterparts in $\mathbb{G}$).
Let $\zeta=T_{\bar G}(\vec x)\rho$, and note that, by the definition of the Kac--Ward matrix, $ \zeta_{\vec e }= \rho_{\vec e}$ for all $\vec e \in \vec{n}(\partial G)$, and $\zeta_{\vec e} = 0$ otherwise.
Therefore, by Lemma~\ref{lem:fermionbound} and the bounded angle property, for any $\vec e =(u,v) \in \text{In}_v$, 
\begin{align*}
\sqrt{2r_u\sin\varepsilon}&\leq \sqrt{2r_u\sin \theta_{\vec e}}\\
 &=   \rho_{\vec e} \\ &= \big[T^{-1}_{\bar G}(\vec x) \zeta\big]_{\vec e} \\& =  \sum_{\vec g \in\vec{n}(\partial G)} T^{-1}_{\bar G}(\vec x)_{\vec e,\vec g}\rho_{\vec g} \\&
  \leq   \sqrt{2R \cos {\varepsilon}}\sum_{\vec g = (w,z) \in\vec{n}(\partial G) }\vec x_{\vec e}\vec x_{-\vec g} \langle  \sigma_v \sigma_w \rangle^{\textnormal{f}}_{G,1} \\&
  \leq  \sqrt{2R \cos {\varepsilon}}\sum_{\vec g = (w,z) \in\vec{n}(\partial G)} \vec x_{-\vec g} \langle  \sigma_v \sigma_w \rangle^{\textnormal{f}}_{G,1} \\
   & \leq  \sqrt{2R \frac{ \cos {\varepsilon}}{ \tan \tfrac{\varepsilon}2}}\varphi_{S,1} (v),
   \end{align*}
 where $r_u$ is the radius of the circle centered at $u$, and
where in the last inequality we used that $\tanh J_{g} = \vec x_{g} =\vec x_{\vec g} \vec x_{-\vec g} \geq \sqrt{ \tan \tfrac{\varepsilon}2} \vec x_{-\vec g}$ by the bounded angle property.
We finish the proof by maximizing over the neighbors $u$ of $v$.
\end{proof}

The last technical statement that we need is an easy consequence of Lemma~\ref{lem:normbound} that will be used to derive exponential decay of correlations for $\beta<1$.
\begin{lemma} \label{lem:normbound1}
Let $\mathbb G^*$ be a circle pattern satisfying the bounded angle property~\eqref{eq:regular}, and
let $\beta>0$. For $e\in E$, let $x_{e}(\beta)= \tanh(\beta J_e)$ where $J_e$ are the coupling constants defined in \eqref{eq:cc}.
For $\vec e\in \vec E$, let 
\[
 \vec x_{\vec e}(\beta) =\sqrt{ \tfrac{x_e(\beta)}{x_e(1)} \tan\tfrac{\theta_{\vec e}}2},
\]
where $e$ is the undirected version of~$\vec e$.
Then there exists a constant $c>0$ such that for all $\tau \in \{\pm 1\}^{\vec E}$ and $\beta \in [0,1]$,
\[
\|\Lambda_{\mathbb G} ( \vec x^{\tau}(\beta))\| \leq 1- c(1-\beta),
\]
where, as before, $\vec x^{\tau}_{\vec e}= \tau_{\vec e} \vec x_{\vec e} $.
\end{lemma}

\begin{proof} 
By the bounded angle property, $M:=\sup_{e\in E} J_e <\infty$.
Define $s(\beta,J_e) =\tfrac{x_e(\beta)}{x_e(1)}= \tfrac{\tanh \beta J_e}{\tanh J_e}$. 
For $ \beta \in [0, 1]$, we have
\[
\frac{\partial }{\partial \beta} s(\beta,J_e) = \frac{J_e}{\tanh J_e \cosh^2 \beta J_e}\geq \frac{2J_e}{\sinh 2J_e} \geq \frac{2M}{\sinh 2M}:=c,
\]
and since $s(1,J_e)=1$, we get $s(\beta,J_e)\leq 1-c(1-\beta)$.
The desired inequality follows directly from Lemma~\ref{lem:normbound}.
\end{proof}

Equipped with the results discussed above, we can now prove our theorems.
\begin{proof}[Proof of Theorem~\ref{thm:main}] We first prove part (i). Let $D$ be the maximal degree of~$\mathbb{G}$, and
let $\vec x (\beta)$ and $x(\beta)$ be as in Lemma~\ref{lem:normbound1}. Note that $\vec x (\beta)$ and $x(\beta)$ are related via \eqref{eq:factorization}.
Let $G$ be a finite subgraph of $\mathbb G$, and denote $T= T_G( \vec x^{\tau}(\beta))$ and $\Lambda =\Lambda_G( \vec x^{\tau}(\beta))$.
{\color{black}We use
the upper bound from Lemma~\ref{lem:fermionbound} and Lemma~\ref{lem:normbound1} to get for $u\neq v$,
\begin{align*}
\langle \sigma_u \sigma_v \rangle_{G,\beta}^f &\leq \sum_{\vec e \in \textnormal{In}_u, \vec g \in \textnormal{Out}_v } \vec x_{-\vec e}\vec x_{\vec g}
{ | T^{-1}_{\vec e, \vec g}|}  \\
& \leq D^2 \max_{\vec e \in \textnormal{In}_u , \vec g \in \textnormal{Out}_v} { | T^{-1}_{\vec e, \vec g}|} \\
&= D^2 \max_{\vec e \in \textnormal{In}_u , \vec g \in \textnormal{Out}_v} {\big | \sum_{k=0}^{\infty} [\Lambda^k]_{\vec e, \vec g} \big|} \\
&= D^2 \max_{\vec e \in \textnormal{In}_u , \vec g \in \textnormal{Out}_v} {\big | \sum_{k=d(u,v)-1}^{\infty} [\Lambda^k]_{\vec e, \vec g} \big|} \\
& \leq  D^2 \frac{ \| \Lambda\|^{d(u,v)-1}}{1 - \| \Lambda\| } \\
& \leq  D^2 \frac{ \| \Lambda_{\mathbb G}( \vec x^{\tau}(\beta))\|^{d(u,v)-1}}{1 - \| \Lambda_{\mathbb G}( \vec x^{\tau}(\beta))\| }\\
& \leq \frac{D^2}{1-\epsilon} \epsilon^{d(u,v)-1},
\end{align*}
where $\epsilon =1- c(1-\beta)<1$ with $c$ as in Lemma~\ref{lem:normbound1}. 
The second equality follows from the fact that $\Lambda$ assigns non-zero transition weights only between adjacent edges and therefore one needs at least $k=d(u,v)-1$ steps for $[\Lambda^k]_{\vec e, \vec g}$ to be non-zero.
The second last inequality holds true since $\Lambda $ is a restriction of $\Lambda_{\mathbb G}( \vec x^{\tau}(\beta))$ to the set of directed edges of~$G$, and hence $ \| \Lambda\| \leq \| \Lambda_{\mathbb G}( \vec x^{\tau}(\beta))\| $.
The final bound is independent of $G$, and therefore we complete the proof of part~(i) by taking the limit of $\langle \sigma_u \sigma_v \rangle_{G,\beta}^f$ as $G\nearrow \mathbb G$.
}

{\color{black}To prove part (ii), denote $g(\beta)=\langle \sigma_v \rangle_{G,\beta}^+ $. We use Lemma~\ref{lem:DumCopTas} together with Lemma~\ref{lem:phi}, and the monotonicity of $\varphi_{S,\beta}(v)$ in $\beta$, to get for all $\beta\geq1$,
\begin{align*}
\frac{d}{d \beta}(1-g(\beta)) & \leq- \tfrac 1 {\beta} \inf_{S \ni v} \varphi_{S,\beta}(v)(1-g(\beta)^2 ) \\
& \leq - \tfrac 1 {\beta} \inf_{S \ni v} \varphi_{S,1}(v)(1-g(\beta)^2 ) \\
&\leq - \frac{t_v}{\beta} (1- g(\beta)^2) \\
& \leq -\frac{t_v}{\beta} (1- g(\beta)),
\end{align*}
where ${t_v} =\sqrt{\frac{ r }{R}} \tan \tfrac{\varepsilon}2 $ is as in Lemma~\ref{lem:phi}.
Using Gr\"{o}nwall's lemma we integrate the resulting differential inequality from $\beta_0=1$ to $\beta$ to get $1-g(\beta) \leq (1-g(1))\beta^{-t_v} \leq \beta^{-t_v}$ since $g(1)\geq 0$. This finishes the proof of part (ii)}
\end{proof}
\begin{proof}[Proof of Theorem~\ref{thm:susc}]
The inequality $\chi_{\mathbb{G},\beta}(v)<\infty$ for $\beta<1$ follows from the exponential decay of correlations from Theorem~\ref{thm:main} (i),
and the quadratic growth of balls in the graph distance (due to the fact that the circles have a minimal diameter).

The fact that $\chi_{\mathbb{G},\beta}(v)=\infty$ for $\beta\geq 1$ follows directly from Lemma~\ref{lem:phi} and the second Griffiths inequality which implies that $\varphi_{S,\beta}(v)$ is non-decreasing in~$\beta$.
\end{proof}

\section{Discrete holomorphicity of fermionic observables} \label{sec:holomorphic}
In this section we discuss discrete holomorphicity of the critical fermionic observable given by the inverse Kac--Ward operator. 
We note that, even though the observable in the general setting of circle patterns satisfies less local constraints than its isoradial version, and hence the tools of~\cite{CheSmi12} are not fully applicable, 
our results point in the direction of conformal invariance of the scaling limit. 

{\color{black}We first give a short overview of the properties of the critical isoradial fermionic observable which were crucial in establishing its scaling limit in~\cite{CheSmi12}.
The \emph{spin fermionic observable} $f$ studied in~\cite{CheSmi12} is related to the complex-valued partition functions $F$ from~\eqref{eq:fo}, and is defined on the midpoints of the edges of a finite subgraph $G$ of an isoradial graph $\mathbb{G}$. At criticality, it satisfies a strong form of discrete holomorphicity called s-holomorphicity, which implies that
\begin{enumerate}
\item the discrete contour integrals of $f$ are well defined both on ${G}$ and ${G}^*$, \label{prop:contig}
\item the real part (or the imaginary part, depending on the precise definitions) of the discrete contour integral of $f^2$, denoted by $h$, is well defined both on ${G}$ and ${G}^*$, \label{prop:contig2}
\item moreover $h$ is sub- and super-harmonic on ${G}$ and ${G}^*$ respectively with respect to the critical Laplacian, \label{prop:harm}
\item $f$ is the unique solution to a certain discrete boundary value problem on~$G$. \label{prop:bc}
\end{enumerate}
We note that these properties are only the starting point for the arguments in~\cite{CheSmi12}, and the full analysis of $f$ in the scaling limit requires many additional technical estimates on the relevant observables.

In this section we will prove that in the case of circle patterns, the corresponding critical fermionic observable satisfies property 
\begin{itemize}
\item \eqref{prop:contig} but only on ${G}^*$ (Lemma~\ref{lem:contin}),
\item \eqref{prop:contig2} also only on ${G}^*$ (the equality from Lemma~\ref{lem:contin2}),
\item \eqref{prop:harm} on $G$ in a \emph{weak sense}. To be more precise, the inequality of Lemma~\ref{lem:contin2} is exactly the same inequality that guarantees that 
the function $h$ is subharmonic on~$G$ in the isoradial case. However, for circle patterns, the function $h$ is not well defined on $G$ as an exact discrete contour integral (for a generalized definition, see~\cite{Chelkak}),
\item \eqref{prop:bc} in the same sense as for isoradial graphs (Corollary~\ref{cor:dRHbvp}).
\end{itemize}
This incomplete picture, compared to the isoradial case, suggests that new ideas are required to carry out the program of~\cite{CheSmi12} in the setting of circle pattterns.
}

To state our results, we first need to recall the notion of {s-holomorhicity} introduced by Smirnov~\cite{smirnov} for the square lattice, and generalized by Chelkak and Smirnov~\cite{CheSmi12} to the setting of isoradial graphs. We assume that $G$ is a finite subgraph of $\mathbb G$ such that $\mathbb G^*$ is a circle pattern.
Let $\text{Proj}(z;\ell)$ be the orthogonal projection of the complex number $z$ onto the complex line~$\ell$.
We say that a function $f: E_G\to \mathbb{C}$ is \emph{s-holomorphic} at an interior vertex $v \in V_G$ if for every dual vertex $v^*$ adjacent to $v$,
\begin{align} \label{eq:shol}
\text{Proj}(f(e_1);  (v-v^*)^{-\frac{1}{2}}\mathbb{R}) = \text{Proj}(f(e_2);  (v-v^*)^{-\frac{1}{2}}\mathbb{R}),
\end{align}
where $e_1$, $e_2$ are the two edges incident on both $v$ and $v^*$, and where, as before, the vertices are identified with the complex numbers given by the embedding.

Our first result will relate the notion of s-holomorphicity to the Kac--Ward matrix. 
To this end, let 
\[
\eta_{\vec e}=h_{\vec e} - t_{\vec e},\quad \text{ and } \quad \ell_{\vec e} = \eta_{\vec e}^{-\frac12} \mathbb{R}.
\] 
We define $\mathcal{L}$ to be the real linear space of functions $\varphi: \vec E \to \mathbb{C}$ 
satisfying $\varphi_{\vec e} \in \ell_{\vec e}$ for all $\vec e$. One can easily check that $\mathcal{L}$ is invariant under the action of the Kac--Ward operator.
Let $S$ be the operator mapping complex functions on $E$ to functions in $\mathcal{L}$ given by
\[
Sf(\vec e) = \rho_{\vec e}  \text{Proj}(f(e);\ell_{\vec e}),
\]
where $\rho_{\vec e}=\sqrt{|e^*|}$ is as before.
Note that $\ell_{\vec e}$ and $\ell_{-\vec e}$ are orthogonal, and hence we have for $\varphi \in \mathcal{L}$,
\[
S^{-1}\varphi( e) = \rho^{-1}_{\vec e}\big( \varphi(\vec e) +\varphi(-\vec e) \big). 
\]
Let $\vec x_{\vec e}= \sqrt{\tan \frac{\theta_{\vec e}}2}$ be the critical weights. The following result in the setting of isoradial graphs was first proved in~\cite{Lis2014}.
\begin{proposition}[s-holomorphicity and the Kac--Ward matrix] \label{thm:sholKW}
A function $f$ is s-holomorphic at an interior vertex $v$ of $G$ if and only if 
\[
T_G(\vec x)Sf(\vec e) = 0 \qquad \text{for all } \vec e \in \textnormal{In}_v.
\]
\end{proposition}
\begin{proof}
Note that $\rho_{(v,u)} = \sqrt{2r_v \sin \theta_{(v,u)} }$, where $r_v$ is the radius of the circle centered at $v$. This implies that
the row of $T_G(\vec x)S$ indexed by $\vec e\in\text{In}_v$ is equal to the corresponding row of the matrix
$TS$ from~\cite{Lis2014} scaled by $2\sqrt{r_v}\vec x_{\vec e}^{-1}$, and hence the result directly follows from Theorem~2.1 of~\cite{Lis2014}. (Note that the matrix $T$ from~\cite{Lis2014}
is our matrix $T(\vec x)$ conjugated by $\text{Diag}\{\vec x_{\vec e}: \vec e \in \vec E \}$).
\end{proof}

This in particular implies that s-holomorphic functions can be uniquely recovered from their boundary values.
Let $G$ and $\bar G$ be as defined before Lemma~\ref{lem:phi},
and let $\varphi\in \mathcal{L}$ be a function defined on the directed edges of $\bar{G}$ and satisfying
\[ 
\varphi(\vec e) \in l_{\vec e} \quad \text{for } \vec e \in \vec{n}(\partial G), \qquad \text{and} \qquad \varphi(\vec e)= 0 \quad \text{otherwise}. 
\]
Following Smirnov~\cite{smirnov}, we say that $f: {E}_{\bar G} \to \mathbb{C}$ solves the \emph{discrete Riemann--Hilbert boundary value problem}
for the pair $(G,\varphi)$ if $f$ is s-holomorphic at all $v\in V_G$ and 
\[
Sf(\vec e) = \varphi(\vec e) \qquad \text{for all }  \vec e \in \vec{n}(G).
\]

\begin{corollary}[Boundary value problem] \label{cor:dRHbvp}
Let $G$ and $\varphi$ be as above and let $T_{\bar G}(\vec x)$ be the critical Kac--Ward operator defined on $\bar{G}$, where 
the half-edges of $\bar{G}$ inherit weights from the corresponding edges of $\mathbb{G}$.
Then the discrete Riemann--Hilbert boundary value problem for $(G,\varphi)$ has a unique solution
\[
 f = S^{-1}T^{-1}_{\bar G}(\vec x) \varphi.
\]
\end{corollary}
\begin{proof}
Suppose that $f$ is a solution to the discrete Riemann--Hilbert boundary value problem. 
Note that $\bar{G}$ is not formally
a subgraph of $\mathbb{G}$ since it contains half-edges. However, these half-edges are parallel to the corresponding edges of $\mathbb{G}$, and therefore we can use Theorem~\ref{thm:sholKW}
to conclude that $T_{\bar G}(\vec x)S f (\vec e) = 0$ for all $\vec e \notin \vec{n}(G)$. 
Moreover if $\vec e \in \vec{n}(G)$, then $h(\vec e)\neq t(\vec g)$ for all $\vec g\neq-\vec e$. Hence by the definition of the Kac--Ward operator, $T_{\bar G}(\vec x)S f(\vec e) = \textnormal{Id} Sf(\vec e)= \varphi(\vec e)$ for $\vec e \in \vec{n}(G)$.
This means that $T_{\bar G}(\vec x)S f(\vec e)= \varphi(\vec e)$ for all directed edges $\vec e$, and the claim follows.
\end{proof}

Finally, we wish to prove that certain discrete contour integrals related to s-holomorphic functions are well defined.
To this end, we first list several basic results.
Fix $f: E \to \mathbb{C}$, and let $\varphi = Sf \in \mathcal{L}$, which means that $f(e)=\rho_{\vec e}^{-1}(\varphi(\vec e)+\varphi(-\vec e))$. 
Define $D_{\vec e} = \eta_{\vec e}/ |\eta_{\vec e} |$. An elementary 
computation yields
\begin{align}
\Re(  \eta_{\vec e^* } f^2(e))& = 2 D_{\vec e^* } \varphi (\vec e) \varphi(-\vec e), \label{eq:Re} \\
 \Im(  \eta_{\vec e^* } f^2(e)) &= |\varphi(\vec e)|^2-|\varphi(-\vec e)|^2 \label{eq:Im}.
\end{align}
Fix a vertex $v\in V$, and let $\varphi_{\text{in}} = (\varphi(-\vec e))_{\vec e\in \text{Out}_v }$ and $\varphi_{\text{out}} = (\varphi(\vec e))_{\vec e\in \text{Out}_v }$.
If $f$ is s-holomorphic, then by Theorem~\ref{thm:sholKW}, $T_G(\vec x)\varphi(\vec e) = 0$ for all $\vec e \in \textnormal{In}_v$. Equivalently,
\begin{align} \label{eq:f2}
\tilde \Lambda \varphi_{\text{out}}  = \varphi_{\text{in}},
\end{align}
where $ \tilde \Lambda= \tilde \Lambda(\vec x)^v$ is the block matrix from Sect.\ \ref{sec:KacWard} acting on the linear subspace indexed by $\text{Out}_v$.
Let $D=\text{Diag}\big\{D_{\vec e}:{\vec e \in \text{Out}_v}\big\}$, and let $\angle(\vec e)= \text{Arg}(D_{\vec e}) \in (-\pi,\pi]$. 
We claim that $D \tilde \Lambda$ is skew-symmetric. Indeed,
for $\vec e, \vec g \in \text{Out}_v$, $\vec e \neq \vec g$, we have
\begin{align*}
\frac{[D\tilde \Lambda]_{\vec e, \vec g}}{[D\tilde \Lambda]_{\vec g, \vec e}} = e^{i(\angle(\vec e)-\angle(\vec g) +\tfrac12\angle(-\vec e, \vec g)  -\tfrac12\angle(-\vec g, \vec e)) }
=e^{i(\angle(\vec e)-\angle(\vec g) +\angle(-\vec e, \vec g))  } 
=  e^{\pm i \pi }
=-1.
\end{align*}

For a directed edge $\vec e$ of $\mathbb G$, let ${\vec e}^*$ be the directed edge of $\mathbb{ G}^*$ crossing $\vec e$ and such that $\angle(\vec e,\vec e^*)>0$.
Recall that $\eta_{\vec e}=h_{\vec e} - t_{\vec e}$.
The next lemma says that the discrete contour integral of an s-holomorphic function is well defined  on $\mathbb G^*$, verifying ``half'' of property~\eqref{prop:contig}.
\begin{lemma}[Vanishing of integrals over closed contours] \label{lem:contin}
If $f$ is s-holomorphic at an interior vertex $v$, then 
\begin{align*}
\sum_{\vec e \in\textnormal{Out}_v} \eta_{\vec e^* } f(e) =0.
\end{align*}
\end{lemma}
\begin{proof}
Let $f= (f(e))_{\vec e \in \text{Out}_v}$. Recall that $\rho=(\sqrt{|e^*|})_{\vec e \in \text{Out}_v}$ is an eigenvector of
$\tilde \Lambda$ of eigenvalue $1$, and let $\rho^2=({|e^*|})_{\vec e \in \text{Out}_v}$. 
The desired equality can be now rewritten as $i (\rho^2)^T Df=0$.
Since $D\tilde \Lambda$ is skew symmetric, we have $D\tilde \Lambda = - \tilde \Lambda^T D$. 
Hence, by \eqref{eq:f2} we get
\begin{align*}
(\rho^2)^T Df &= \rho^TD (\varphi_{\text{out}}+\varphi_{\text{in}}) \\
 &=  \rho^T  D\varphi_{\text{out}} + \rho^T  D \tilde \Lambda \varphi_{\text{out}} \\
 &=\rho^T  D \varphi_{\text{out}}   -\rho^T \tilde  \Lambda^T D  \varphi_{\text{out}}  \\
 &= \rho^T  D \varphi_{\text{out}}  -( \tilde \Lambda \rho)^T  D  \varphi_{\text{out}} \\
 &= 0,
\end{align*}
which completes the proof.
\end{proof}

We finish with a result saying that if $f$ is s-holomorphic, then the real part of the discrete contour integral of $f^2$ is well defined on $\mathbb{G}^*$, and moreover, the imaginary part of the integral of $f^2$ over any closed counterclockwise 
contour on $\mathbb{G}^*$ is nonnegative. This is the counterpart in the setting of circle patterns of the fundamental properties~\eqref{prop:contig2} and \eqref{prop:harm} discovered by Smirnov~\cite{smirnov}, and Chelkak and Smirnov~\cite{CheSmi12}.

\begin{lemma} \label{lem:contin2}
If $f$ is s-holomorphic at an interior vertex $v$, then 
\begin{align*}
\Re\Big(\sum_{\vec e \in\textnormal{Out}_v} \eta_{\vec e^* } f^2(e)\Big) =0, \quad \text{and} \quad  \Im \Big(\sum_{\vec e \in\textnormal{Out}_v} \eta_{\vec e^* } f^2(e)\Big) \geq 0.
\end{align*}
\end{lemma}
\begin{proof}

By \eqref{eq:Re} and \eqref{eq:f2}, we have
\[
\Re \Big(\sum_{\vec e \in\textnormal{Out}_v} \eta_{\vec e^* } f^2(e)\Big) =2i\varphi_{\text{out}}^TD\varphi_{\text{in}}= 2i \varphi_{\text{out}}^TD\tilde \Lambda\varphi_{\text{out}}= 0,
\]
where in the last equality we used that $D\tilde\Lambda$ is skew-symmetric.
 
Moreover, by \eqref{eq:Im}, \eqref{eq:f2} and the fact that the operator norm of $\tilde \Lambda$ is bounded by $1$ \eqref{eq:lisnorm}, we get
\[
\Im \Big(\sum_{\vec e \in\textnormal{Out}_v} \eta_{\vec e^* } f^2(e)\Big) =\|\varphi_{\text{out}}\|^2- \|\varphi_{\text{in}}\|^2 =\|\varphi_{\text{out}}\|^2 -\|\tilde \Lambda \varphi_{\text{out}}\|^2 \geq 0,
\]
which completes the proof.
\end{proof}

\begin{remark}
Note that the last two lemmas follow directly from the corresponding results of Chelkak and Smirnov for isoradial graphs~\cite{CheSmi12}.
Indeed, both the definition of s-holomorphicity \eqref{eq:shol} and the desired relations depend only on the geometry of the graph $\mathbb{G}^*$ in the immediate neighborhood of~$v$, 
which is indistinguishable from the one of an isoradial graph.
However, we included the concise proofs that use the Kac--Ward matrix as they shed a different light on these relations.  
\end{remark}

{\color{black}\section{Applications to other graphs} \label{sec:other}
In this section we briefly discuss the consequences of a non-trivial kernel of the Kac--Ward matrix for the question of criticality of the Ising model defined on other types of planar graphs.
\subsection*{Biperiodic graphs}
Cimasoni and Duminil-Copin~\cite{CimDum} (see also~\cite{Li2012}) computed the critical temperature of the the Ising model on an arbitrary biperiodic graph, i.e., a graph which is invariant under 
the action of a $\mathbb{Z}^2$-isomorphic group of translations of the plane. The proof uses the Kac--Ward matrix, and the critical point is identified with the only inverse-temperature $\beta_c$ at which 
there exists a non-trivial periodic vector in the kernel of this matrix. 
Moreover it was shown that at $\beta_c$ the Kac--Ward matrix associated with the \emph{dual} Ising model also has a non-trivial kernel.

Using this one can alternatively obtain criticality of $\beta_c$ without going through the analysis of differentiability of the free energy, as it is done in~\cite{CimDum}.
Indeed, the existence at~$\beta_c$ of the null vectors of the primal and dual Kac--Ward matrix implies a direct analog of Lemma~\ref{lem:phi} which gives a non-zero lower bound on the quantity~$\varphi$ from \eqref{eq:phi} for both the primal and dual model.
Hence, together with the differential inequality from Lemma~\ref{lem:DumCopTas}, exactly as in the proof of part (ii) of Theorem~\ref{thm:main}, we get that the magnetization in the primal Ising model is positive for $\beta >\beta_c$, 
and in the dual model it is positive for $\beta<\beta_c$.
To get a complete picture, one needs to show that the magnetization is never simultaneously positive in both the primal and dual model. This is e.g.\ proved in Theorem~4.4 of~\cite{CimDum}, which 
also gives exponential decay of the two-point spin correlation functions under the assumption that the magnetization vanishes.

\subsection*{S-embeddings} In~\cite{Chelkak}, Chelkak introduced a family of Ising models defined on s-embeddings, which are more general that the critical models on circle patterns. 
S-embeddings are planar graphs (or more precisely pairs of primal and dual graphs) defined by the property that each quad whose diagonals are given by a pair of a primal and dual edge is tangential. This is clearly satisfied for circle patterns since all such quads are in this case kites.
A natural question is if the Ising model of Chelkak is critical in the sense of spin correlations and magnetization.
The answer in the biperiodic case is affirmative and is given by the main result of~\cite{CimDum}, 
and the case of circle patterns with uniformly bounded angles and faces is covered in the present article. 

In the general setting a partial answer, though strongly supporting criticality, can be provided using the same idea as above, i.e. by constructing a null-vector both for the primal and dual Kac--Ward matrix.

\begin{remark}
The fact that the Kac--Ward matrix on s-embeddings with weights as in~\cite{Chelkak} has a non-trivial kernel was first observed by Dmitry Chelkak and the author is grateful to him for sharing his insight.
\end{remark}
We now give the details of the construction, following the setup of~\cite{Chelkak}.
Let $\mathbb{G}$ and $\mathbb{G}^*$ form an s-embedding. A \emph{corner} $c=\{u,v\}$ of $\mathbb G$ is a pair composed of a vertex~$v$ and a face $u$ incident on $v$. 
We define $\Upsilon$ to be the 4-regular graph whose vertex set is the set of all corners of $\mathbb G$, and where two corners $\{u_1,v_1\}$, $\{u_2,v_2\}$ are adjacent if either $u_1=u_2$ and $\{v_1,v_2\}$ is a primal edge, or $v_1=v_2$ and $\{u_1,u_2\}$ is a dual edge (see e.g.\ Fig.\ 26 in~\cite{Mercat}).
We denote by $\Upsilon^{\times}$ the double cover of $\Upsilon$ that branches around each of its faces (see Fig.\ 27 in~\cite{Mercat}).

A \emph{spinor} is a function defined on the vertices of $\Upsilon^{\times}$ whose values at two corners corresponding to one corner in $\Upsilon$ differ by the sign.
Chelkak defined a spinor $\mathcal{F}(c) := (v-u)^{ 1 /2}$ (Remark 6.2 of~\cite{Chelkak}), where $c$ is one of the two versions of the corner $\{u,v\}$, and where the face $u$ is identified with the complex number given by the embedding of $\mathbb{G}^*$.
It turns out that $\mathcal{F}$ satisfies the \emph{three-term spinor relation} (equation (2.13) of~\cite{Chelkak}). To state it, consider a face of $\Upsilon$ corresponding to an edge $e$ of $\mathbb{G}$, and let $c_1, c_2, c_3$ be any three consecutive corners in $\Upsilon^{\times}$ as one goes around the face counterclockwise. Let $x_e =\tan \tfrac \theta 2$ be the 
weight associated to $e$ in the Ising model.
The three-term relation reads 
\begin{align} \label{eq:threeterm}
\mathcal{F}(c_2) = p_e \mathcal{F}(c_1)  + \sqrt{1-p_e^2}\mathcal{F}(c_3),
\end{align}
where $p_e=\sin \theta$ if the corner $c_2=\{u,v\}$ corresponds to the directed edge $(v,u)$ in the tangential quad assigned to $e$ when going counterclockwise around $e$, 
and $p_e=\cos \theta$ if it corresponds to $(u,v)$.

The crucial observation now is that the kernels of the three-term relation and the Kac--Ward matrix are related to each other as was shown in~\cite{CCK}. To be more precise, 
note that $|\mathcal{F}(c)|=|v-u|^{1/2} $ is a well defined function on the corners in $\Upsilon$ that satisfies the three term relation~\eqref{eq:threeterm} with coefficients multiplied by $\bar \eta(c):= \mathcal{F}(c)/ {|\mathcal{F}(c)|}$. This means that $|\mathcal{F}|$ is in the kernel of the (infinite) matrix $D$ defined before Lemma~3.4 in~\cite{CCK}.
We easily obtain from this lemma that $2D=C[I-\tfrac12(Y+iI)D]$, where~$I$ is the identity, $Y$ is the involutive automorphism induced by $\vec e\mapsto -\vec e$, 
and where $C$ is an explicit conjugate of the (infinite) Kac--Ward matrix defined after equation (3.1) in~\cite{CCK}.
Hence, $|\mathcal{F}|$ is the kernel of $C$, and therefore $\rho:=B^{*}|\mathcal{F}| $ is in the kernel of the Kac--Ward matrix as defined in~\cite{CCK}, where $B$ is the (infinite) block-diagonal matrix defined in equation (3.1) of~\cite{CCK}. 
By the same construction applied to $\mathbb{G}^*$ we also obtain a null-vector $\rho^*$ of the dual Kac--Ward matrix. 

\begin{remark}
The Kac--Ward matrix in~\cite{CCK} is the same as ours, up to conjugation by a diagonal matrix with diagonal terms given by $\big(\tfrac{\vec x_{\vec e}}{\vec x_{-\vec e}}\big)^{1/2}$.
We note that can use the construction above to obtain in an alternative way the null-vector for circle-patterns from Proposition~\ref{thm:eigenvector}.
\end{remark}

It is easily seen from the definitions of $B$ and $|\mathcal{F}|$ that if the half-angles of the tangential quads associated to $\mathbb{G}$ and $\mathbb{G}^*$ satisfy condition~\eqref{eq:anglecond}, and moreover the sizes of the quads are uniformly bounded from above, then $\|\rho\|_{\infty}<\infty$ and $\|\rho^*\|_{\infty}<\infty$. By the same arguments as in the biperiodic case, this means that for $\beta>1$, there is positive magnetization in the primal Ising model of Chelkak, and for $\beta<1$ there is positive magnetization in the associated dual model. This is a strong indication that the model at $\beta=1$ is indeed critical.
 
However, it is not clear how to show that the magnetization cannot be simultaneously positive in both the primal and dual Ising model in a general non-periodic s-embedding.
Indeed, the arguments of Theorem~4.4 of~\cite{CimDum} and the operator norm estimates from the present article do not apply in this more general setting.
We leave this as an open problem.
}
\appendix
\include{appendix_name}
\section{} \label{app:a}
The following proof of Lemma~\ref{lem:DumCopTas} is due to Duminil-Copin and Tassion, and (up to small modifications) 
is contained in an unpublished version of~\cite{DumCopTas} (proof of Lemma 2.4 in~\cite{DumCopTasU}). We include it here for completeness.
The graphs considered in this section are not assumed to be planar.

Let $G=(V_G,E_G)$ be a finite graph, and let $\partial V_G \subset V_G$ be a fixed set of vertices called the \emph{boundary}. 
We consider an augmented graph $G_{\bullet}=(V_{\bullet},E_{\bullet})$ where a \emph{ghost} (or a \emph{boundary}) vertex $\bullet$ is added to the vertex set,
and edges of the form $\{\bullet , v\}$, $v\in \partial V_G$, are added to the edge set.

We will consider Ising models on subgraphs of $G_{\bullet}$ as defined in Sect.~\ref{sec:main} with arbitrary positive coupling constants $(J_e)_{e \in  E_{\bullet}}$
(to recover the exact setting of Sect.~\ref{sec:main}, one needs to set $J_{v,\bullet}=\sum_{u \notin V_G } J_{\{u,v\}}$ for $v\in\partial V_G$).
One can easily check that the following relation between boundary conditions holds true for any $v\in V_G$,
\[
\langle \sigma_v \rangle_{G,\beta}^+ = \langle \sigma_v \sigma_{\bullet} \rangle_{G_{\bullet},\beta}^f.
\]
For $S\subset  V_{\bullet}$, we will write $\langle \cdot \rangle_{S,\beta}^f$ for the expectation of the Ising model with free boundary conditions
defined on the subgraph of $G_{\bullet}$ induced by $S$.

The notion of a {current} will play a crucial role in the proof. A \emph{current} is a function $\n: E_{\bullet} \to \mathbb{N}$, whose \emph{weight} is given by 
\[
w(\n) = \prod_{e\in E_{\bullet}} \frac{(\beta J_{e})^{\n_e}}{ (\n_{e})!}.
\]
A vertex $v$ is called a \emph{source} of a current $\n$ if $\sum_{u\sim v} \n_{\{v,u\}}$ is odd.
We denote by $\partial \n$ the set of all sources of $\n$. 
The classical \emph{random current representation} of correlation functions due to Griffiths, Hurst and Sherman~\cite{GHS} yields for any $A\subset S \subset V_{\bullet}$,
\[ \displaystyle
 \big \langle \prod_{v\in A} \sigma_v \big\rangle_{S,\beta}^f = \frac{\sum_{\partial \n = A } w(\n)}{\sum_{\partial \n = \emptyset } w(\n)},
\]
where both sums are over currents that are zero outside the subgraph induced by $S$.
For $u,v \in V_{\bullet}$ and a current $\n$, we will write $u\con[\n] v$ if $u$ is connected to~$v$ via a path of edges in $E_{\bullet}$ with nonzero values of $\n$,
and $u\ncon[\n] v$ if it is not connected.
The last tool that we will need is the celebrated \emph{switching lemma} introduced in~\cite{GHS} and developed by Aizenman in~\cite{aizenman}.
It says that if $A \cup\{ u,v\}\subset S \subset V_{\bullet}$, then
  \begin{equation*}
    \sum_{\substack{\partial \n_1= A \Delta \{u,v\}\\ \partial \n_2=
        \{u,v\}}}w(\n_1) w(\n_2)=\sum_{\substack{\partial \n_1=
        A\\ \partial \n_2=\emptyset}}w(\n_1) w(\n_2) \mathbf{1}[u\con[\n] v],
  \end{equation*}
where $\n=\n_1+\n_2$, $\Delta$ is the symmetric difference, and where again the sums are taken over currents $\n_1$, $\n_2$ that are zero outside the subgraph induced by $S$.

\begin{proof}[Proof of Lemma~\ref{lem:DumCopTas}] Fix a vertex $0\in V$.
From the definition of the Ising model, we have
\begin{equation*}
  \frac{d}{d\beta}\langle
  \sigma_0\rangle_{G,\beta}^+=\frac{d}{d\beta}\langle
  \sigma_0\sigma_{\bullet}\rangle_{G_{\bullet},\beta}^f=\hspace{-0.3cm}\sum_{\{u,v\} \in E_{\bullet}}\hspace{-0.3cm}J_{\{u,v\}}\big(\langle
  \sigma_0\sigma_{ \bullet}\sigma_u\sigma_v \rangle_{G_{\bullet},\beta}^f -\langle
  \sigma_0 \sigma_{ \bullet} \rangle_{G_{\bullet},\beta}^f \langle
  \sigma_u\sigma_v \rangle_{G_{\bullet},\beta}^f\big).
\end{equation*}

Let $Z_{G_{\bullet},\beta}=\sum_{ \partial \n =\emptyset} w(\n)$ be the partition function of currents on $G_{\bullet}$ with no sources.
Using the random current representation of correlation functions and the switching
lemma for $S=V_{\bullet}$, we obtain
\begin{equation*}
 \langle \sigma_0\sigma_{ \bullet}\sigma_u\sigma_v \rangle_{G_{\bullet},\beta}^f -\langle
  \sigma_0 \sigma_{ \bullet}\rangle_{G_{\bullet},\beta}^f \langle
  \sigma_u\sigma_v \rangle_{G_{\bullet},\beta}^f  \\
  = \frac{1}{Z_{G_{\bullet},\beta}^2}\sum_{\substack{\partial
      \n_1=\{0, \bullet\}\Delta\{u,v\} \\ \partial\n_2=\emptyset}}\hspace{-0.3cm}w(\n_1)w(\n_2)\mathbf
  1[0\ncon[\n] \bullet],
\end{equation*}
where $\n=\n_1+\n_2$.
Note that if $\n_1=\{0, \bullet\}\Delta\{u,v\}$, $\partial\n_2=\emptyset$ and
      $0\ncon[\n] \bullet$, then either $0\con[\n]u$ and $v\con[\n] \bullet$, or $0\con[\n]v$ and
        $u\con[\n] \bullet$.
Since the second case is the same as the first one with $u$ and
$v$ exchanged, we get 
\begin{equation}
  \label{eq:12}
   \frac{d}{d\beta}\langle
  \sigma_0\rangle_{ G,\beta}^+=\frac{1}{Z_{G_{\bullet},\beta}^2}
  \sum_{u\in V_{\bullet}}\sum_{v\in V_{\bullet}} J_{\{u,v\}}\delta_{u,v},  
\end{equation}
where 
\[
\delta_{u,v}=\displaystyle\sum_{\substack{\partial
      \n_1=\{0, \bullet\}\Delta\{u,v\}\\ \partial\n_2=\emptyset}}w(\n_1)w(\n_2)\mathbf
  1[0\con[\n]u, v\con[\n] \bullet, 0\ncon[\n] \bullet].
  \]

For $z\in \{0, \bullet\}$, 
define $\mathcal{S}_z$ to be the set of vertices in $V_{\bullet}$ that are \emph{not connected} to $z$ in $\n$.
Let us compute $\delta_{u,v}$ by summing over all possible $\mathcal{S}_0$:
\begin{align*}
  \delta_{u,v}&=\sum_{\substack{S\subset  V_{\bullet}}}\ \sum_{\substack{\partial
      \n_1=\{0, \bullet\}\Delta\{u,v\}\\ \partial\n_2=\emptyset}}w(\n_1)w(\n_2)\mathbf
  1[\mathcal{S}_0=S, 0\con[\n]u, v\con[\n] \bullet, 0\ncon[\n] \bullet]\\
  &=\sum_{\substack{S\subset  V_{\bullet}\\ \text{s.t. }v, \bullet \in
      S\\
      \text{and }0,u\notin S}}\ \sum_{\substack{\partial
      \n_1=\{0, \bullet\}\Delta\{u,v\}\\ \partial\n_2=\emptyset}}w(\n_1)w(\n_2)\mathbf
  1[\mathcal{S}_0=S, v\con[\n] \bullet].
\end{align*}
Note that when $\mathcal{S}_0=S$, $\n_1$ and $\n_2$ 
vanish on every $\{x,y\}$ with $x\in S$ and $y\notin S$. Thus, for $i=1,2$, we can decompose $\n_i$ as $\n_i=\n_i^S+\n_i^{ V_{\bullet}\setminus S}$, where $\n_i^A$
denotes the current with zero values outside the subgraph induced by $A$, and with sources $\partial\n_i^A=A\cap \partial
\n_i$. Together with the last equality and the random current representation of correlation functions, this gives
\begin{equation*}
 \delta_{u,v}=\sum_{\substack{S\subset  V_{\bullet}\\ \text{s.t. }v, \bullet \in
      S\\
      \text{and }0,u\notin S}}\ \sum_{\substack{\partial
      \n_1=\{0\}\Delta\{u\}\\ \partial\n_2=\emptyset}}w(\n_1)w(\n_2)\langle\sigma_v \sigma_{\bullet}\rangle_{S,\beta}^f\mathbf
  1[\mathcal{S}_0=S].
\end{equation*}
Since $1\geq \langle\sigma_v\sigma_{\bullet}\rangle_{S,\beta}^f$, we get
\begin{align*}
 \delta_{u,v} &\ge \sum_{\substack{S\subset  V_{\bullet}\\ \text{s.t. }v, \bullet \in
      S\\
      \text{and }0,u\notin S}}\ \sum_{\substack{\partial
      \n_1=\{0\}\Delta\{u\}\\ \partial\n_2=\emptyset}}w(\n_1)w(\n_2)(\langle\sigma_v\sigma_{\bullet}\rangle_{S,\beta}^f)^2\mathbf
  1[\mathcal{S}_0=S]\\
&= \sum_{\substack{S\subset  V_{\bullet}\\ \text{s.t. }v, \bullet \in
      S\\
      \text{and }0,u\notin S}}\ \sum_{\substack{\partial
      \n_1=\{0\}\Delta\{u\}\Delta\{v, \bullet\}\\ \partial\n_2=\{v, \bullet\}}}w(\n_1)w(\n_2)\mathbf
  1[\mathcal{S}_0=S]\\
&=\sum_{\substack{S\subset  V_{ \bullet}\\ \text{s.t. }v, \bullet \in
      S\\
      \text{and }0,u\notin S}}\ \sum_{\substack{\partial
      \n_1=\{0\}\Delta\{u\}\\ \partial\n_2=\emptyset}}w(\n_1)w(\n_2)\mathbf
  1[\mathcal{S}_0=S, v\con[\n] \bullet]\\
&=\sum_{\substack{\partial
      \n_1=\{0\}\Delta\{u\}\\ \partial\n_2=\emptyset}}w(\n_1)w(\n_2)\mathbf
  1[v\con[\n] \bullet,0\ncon[\n] \bullet],
\end{align*}
where in the second line we used the random current representation of correlation functions, and in the third line again the switching lemma. 
We now sum over all possible $\mathcal{S}_{ \bullet}$:
\begin{align*}\delta_{u,v}&\ge \sum_{S\subset V_G}\ \sum_{\substack{\partial
      \n_1=\{0\}\Delta\{u\}\\ \partial\n_2=\emptyset}}w(\n_1)w(\n_2)\mathbf
  1[\mathcal{S}_{ \bullet}=S,v\con[\n] \bullet, 0\ncon[\n] \bullet]\\
  &= \sum_{\substack{S\subset V_G\\
    \text{s.t. }0,u\in S\\\text{and }v \in V_{\bullet} \setminus S}}\ \sum_{\substack{\partial
      \n_1=\{0\}\Delta\{u\}\\ \partial\n_2=\emptyset}}w(\n_1)w(\n_2)\mathbf
  1[\mathcal{S}_{ \bullet}=S]\\
  &= \sum_{\substack{S\subset V_G\\
    \text{s.t. }0,u\in S\\ \text{and }v \in V_{\bullet} \setminus S}}\ \sum_{\substack{
      \partial \n_1=\emptyset\\ \partial\n_2=\emptyset}}w(\n_1)w(\n_2)\langle\sigma_0\sigma_u\rangle_{S,\beta}^{\rm f}\mathbf  1[\mathcal{S}_{ \bullet}=S].
\end{align*}
The third line follows from the fact that since $\mathcal{S}_{ \bullet}=S$, $\n_1$ and $\n_2$ can be decomposed
as $\n_i=\n_i^S+\n_i^{ V_{\bullet}\setminus S}$ as before. 

Combining
the inequality above with \eqref{eq:12}, we get
\begin{align*}
    \frac{d}{d\beta}\langle
  \sigma_0\rangle_{ G,\beta}^+ &\geq\frac1{Z_{G_{\bullet},\beta}^2}
 \sum_{\substack{S\subset V_G\\ S\ni 0}} \sum_{u\in S}\sum_{v \in V_{\bullet} \setminus S} \sum_{\substack{
      \partial \n_1=\emptyset\\ \partial\n_2=\emptyset}}w(\n_1)w(\n_2)J_{u,v}\langle\sigma_0\sigma_u\rangle_{S,\beta}^{\rm f}\mathbf  1[\mathcal{S}_{\bullet}=S]\\
        &\ge\frac1\beta \sum_{\substack{S\subset V
\\ S\ni 0}}\varphi_{S,\beta} (0)  \frac1{Z_{G_{\bullet},\beta}^2} \sum_{\substack{\partial
      \n_1=\emptyset\\ \partial\n_2=\emptyset}}w(\n_1)w(\n_2)\mathbf
  1[\mathcal{S}_{ \bullet}=S]\label{eq:13}\\
  &\ge \frac1{\beta} \inf_{S\ni 0}\varphi_{S,\beta} (0) \frac1{Z_{G_{\bullet},\beta}^2}{\displaystyle \sum_{\substack{
      \partial \n_1=\emptyset\\ \partial\n_2=\emptyset}}w(\n_1)w(\n_2)\big(1-\mathbf
  1[0\con[\n] \bullet]\big)}\nonumber\\
            &=\frac1\beta \inf_{S\ni 0}\varphi_{S,\beta} (0)(1-(\langle
  \sigma_0\sigma_{\bullet}\rangle_{ G_{\bullet},\beta}^f)^2)\nonumber\\
      &=\frac1\beta \inf_{S\ni 0}\varphi_{S,\beta} (0)(1-(\langle
  \sigma_0\rangle_{ G,\beta}^+)^2),\nonumber
\end{align*}
where in the second inequality we used that $\beta J_{\{u,v \}} \geq \tanh (\beta J_{\{ u,v\}})$, and in the first equality we used the random current representation of correlations and the switching lemma for the last time.
\end{proof}

\bibliographystyle{amsplain}
\bibliography{circlepattern}

\begin{bibdiv}
\begin{biblist}

\bib{Bobenko1}{article}{
      author={Agafonov, S.~I.},
      author={Bobenko, A.~I.},
       title={Discrete {$Z^\gamma$} and {P}ainlev\'e equations},
        date={2000},
        ISSN={1073-7928},
     journal={Internat. Math. Res. Notices},
      number={4},
       pages={165\ndash 193},
         url={https://doi.org/10.1155/S1073792800000118},
      review={\MR{1747617}},
}

\bib{aizenman}{article}{
      author={Aizenman, M.},
       title={{Geometric analysis of $\varphi^{4}$ fields and Ising models. I,
  II}},
        date={1982},
     journal={Comm. Math. Phys.},
      volume={86},
      number={1},
       pages={1\ndash 48},
         url={http://projecteuclid.org/euclid.cmp/1103921614},
}

\bib{Baxter}{article}{
      author={Baxter, R.~J.},
       title={Free-fermion, checkerboard and {${\bf Z}$}-invariant lattice
  models in statistical mechanics},
        date={1986},
        ISSN={0962-8444},
     journal={Proc. Roy. Soc. London Ser. A},
      volume={404},
      number={1826},
       pages={1\ndash 33},
      review={\MR{836281}},
}

\bib{Bobenko2}{article}{
      author={Bobenko, Alexander~I.},
      author={Hoffmann, Tim},
       title={Hexagonal circle patterns and integrable systems: patterns with
  constant angles},
        date={2003},
        ISSN={0012-7094},
     journal={Duke Math. J.},
      volume={116},
      number={3},
       pages={525\ndash 566},
         url={https://doi.org/10.1215/S0012-7094-03-11635-X},
      review={\MR{1958097}},
}

\bib{Bobenko3}{article}{
      author={Bobenko, Alexander~I.},
      author={Mercat, Christian},
      author={Suris, Yuri~B.},
       title={Linear and nonlinear theories of discrete analytic functions.
  {I}ntegrable structure and isomonodromic {G}reen's function},
        date={2005},
        ISSN={0075-4102},
     journal={J. Reine Angew. Math.},
      volume={583},
       pages={117\ndash 161},
         url={https://doi.org/10.1515/crll.2005.2005.583.117},
      review={\MR{2146854}},
}

\bib{BCL}{article}{
      author={Bonzom, V.},
      author={Costantino, F.},
      author={Livine, E.~R.},
       title={Duality between spin networks and the 2{D} {I}sing model},
        date={2016},
        ISSN={0010-3616},
     journal={Comm. Math. Phys.},
      volume={344},
      number={2},
       pages={531\ndash 579},
         url={https://doi.org/10.1007/s00220-015-2567-6},
      review={\MR{3500249}},
}

\bib{BdT1}{article}{
      author={Boutillier, C.},
      author={de~Tili\`ere, B.},
       title={The critical {${\bf Z}$}-invariant {I}sing model via dimers: the
  periodic case},
        date={2010},
        ISSN={0178-8051},
     journal={Probab. Theory Related Fields},
      volume={147},
      number={3-4},
       pages={379\ndash 413},
         url={https://doi.org/10.1007/s00440-009-0210-1},
      review={\MR{2639710}},
}

\bib{BdT2}{article}{
      author={Boutillier, C.},
      author={de~Tili\`ere, B.},
       title={The critical {$Z$}-invariant {I}sing model via dimers: locality
  property},
        date={2011},
        ISSN={0010-3616},
     journal={Comm. Math. Phys.},
      volume={301},
      number={2},
       pages={473\ndash 516},
         url={https://doi.org/10.1007/s00220-010-1151-3},
      review={\MR{2764995}},
}

\bib{BdTR}{unpublished}{
      author={Boutillier, C.},
      author={de~Tili\`ere, B.},
      author={K., Raschel},
       title={{The Z-invariant Ising model via dimers}},
        date={2016},
        note={arXiv:1612.09082},
}

\bib{Chelkak}{unpublished}{
      author={Chelkak, D.},
       title={{Planar Ising model at criticality: state-of-the-art and
  perspectives}},
        date={2017},
        note={arXiv:1712.04192},
}

\bib{CCK}{unpublished}{
      author={Chelkak, D.},
      author={Cimasoni, D.},
      author={Kassel, A.},
       title={{Revisiting the combinatorics of the 2D Ising model}},
         how={private communication},
        date={2015},
        note={to appear in {Ann. Inst. Henri Poincar{\'e} Comb. Phys.
  Interact.}},
}

\bib{CheSmi12}{article}{
      author={Chelkak, Dmitry},
      author={Smirnov, Stanislav},
       title={Universality in the 2{D} {I}sing model and conformal invariance
  of fermionic observables},
        date={2012},
        ISSN={0020-9910},
     journal={Invent. Math.},
      volume={189},
      number={3},
       pages={515\ndash 580},
         url={http://dx.doi.org/10.1007/s00222-011-0371-2},
      review={\MR{2957303}},
}

\bib{Cimasoni2015}{article}{
      author={Cimasoni, David},
       title={Kac-{W}ard operators, {K}asteleyn operators, and s-holomorphicity
  on arbitrary surface graphs},
        date={2015},
        ISSN={2308-5827},
     journal={Ann. Inst. Henri Poincar\'e D},
      volume={2},
      number={2},
       pages={113\ndash 168},
         url={http://dx.doi.org/10.4171/AIHPD/16},
      review={\MR{3354329}},
}

\bib{CimDum}{article}{
      author={Cimasoni, David},
      author={Duminil-Copin, Hugo},
       title={The critical temperature for the {I}sing model on planar doubly
  periodic graphs},
        date={2013},
        ISSN={1083-6489},
     journal={Electron. J. Probab.},
      volume={18},
       pages={no. 44, 18},
         url={http://dx.doi.org/10.1214/EJP.v18-2352},
      review={\MR{3040554}},
}

\bib{DolEtAl}{article}{
      author={Dolbilin, N.~P.},
      author={Zinov'ev, Yu.~M.},
      author={Mishchenko, A.~S.},
      author={Shtan'ko, M.~A.},
      author={Shtogrin, M.~I.},
       title={The two-dimensional {I}sing model and the {K}ac-{W}ard
  determinant},
        date={1999},
        ISSN={0373-2436},
     journal={Izv. Ross. Akad. Nauk Ser. Mat.},
      volume={63},
      number={4},
       pages={79\ndash 100},
      review={\MR{1717680 (2000j:82008)}},
}

\bib{Duffin}{article}{
      author={Duffin, R.~J.},
       title={Potential theory on a rhombic lattice},
        date={1968},
        ISSN={0021-9800},
     journal={Journal of Combinatorial Theory},
      volume={5},
      number={3},
       pages={258 \ndash  272},
  url={http://www.sciencedirect.com/science/article/pii/S0021980068800729},
}

\bib{DumCopTasU}{unpublished}{
      author={Duminil-Copin, H.},
      author={Tassion, V.},
       title={{A new proof of the sharpness of the phase transition for
  Bernoulli percolation and the Ising model}},
        date={2015},
        note={arXiv:1502.03050v1},
}

\bib{DumCopTas}{article}{
      author={Duminil-Copin, Hugo},
      author={Tassion, Vincent},
       title={A new proof of the sharpness of the phase transition for
  {B}ernoulli percolation and the {I}sing model},
        date={2016},
        ISSN={0010-3616},
     journal={Comm. Math. Phys.},
      volume={343},
      number={2},
       pages={725\ndash 745},
         url={https://doi.org/10.1007/s00220-015-2480-z},
      review={\MR{3477351}},
}

\bib{Fisher}{article}{
      author={Fisher, M.~E.},
       title={On the dimer solution of planar ising models},
        date={1966},
     journal={Journal of Mathematical Physics},
      volume={7},
      number={10},
       pages={1776\ndash 1781},
      eprint={https://doi.org/10.1063/1.1704825},
         url={https://doi.org/10.1063/1.1704825},
}

\bib{GHS}{article}{
      author={Griffiths, R.~B.},
      author={Hurst, C.~A.},
      author={Sherman, S.},
       title={{Concavity of Magnetization of an Ising Ferromagnet in a Positive
  External Field}},
        date={1970},
     journal={Journal of Mathematical Physics},
      volume={11},
      number={3},
       pages={790\ndash 795},
  url={http://scitation.aip.org/content/aip/journal/jmp/11/3/10.1063/1.1665211},
}

\bib{Ising}{article}{
      author={Ising, E.},
       title={Beitrag zur {T}heorie des {F}erromagnetismus},
        date={1925FEB--APR},
        ISSN={0044-3328},
     journal={Z. Physik},
      volume={31},
       pages={253\ndash 258},
}

\bib{KacWard}{article}{
      author={Kac, M.},
      author={Ward, J.~C.},
       title={A combinatorial solution of the two-dimensional {I}sing model},
        date={1952},
        ISSN={0031-899X},
     journal={Phys. Rev.},
      volume={88},
      number={6},
       pages={1332\ndash 1337},
}

\bib{KLM}{article}{
      author={Kager, Wouter},
      author={Lis, Marcin},
      author={Meester, Ronald},
       title={{The Signed Loop Approach to the Ising Model: Foundations and
  Critical Point}},
        date={2013Jul},
        ISSN={1572-9613},
     journal={Journal of Statistical Physics},
      volume={152},
      number={2},
       pages={353\ndash 387},
         url={https://doi.org/10.1007/s10955-013-0767-z},
}

\bib{KraWan1}{article}{
      author={Kramers, H.~A.},
      author={Wannier, G.~H.},
       title={Statistics of the two-dimensional ferromagnet.~{I}},
        date={1941},
     journal={Phys. Rev. (2)},
      volume={60},
       pages={252\ndash 262},
      review={\MR{0004803 (3,63i)}},
}

\bib{Li2012}{article}{
      author={Li, Zhongyang},
       title={{Critical Temperature of Periodic Ising Models}},
        date={2012Oct},
        ISSN={1432-0916},
     journal={Communications in Mathematical Physics},
      volume={315},
      number={2},
       pages={337\ndash 381},
         url={https://doi.org/10.1007/s00220-012-1571-3},
}

\bib{Lis2014a}{article}{
      author={Lis, Marcin},
       title={{Phase Transition Free Regions in the Ising Model via the
  Kac--Ward Operator}},
        date={2014Nov},
        ISSN={1432-0916},
     journal={Communications in Mathematical Physics},
      volume={331},
      number={3},
       pages={1071\ndash 1086},
         url={https://doi.org/10.1007/s00220-014-2061-6},
}

\bib{Lis2014}{article}{
      author={Lis, Marcin},
       title={{The Fermionic Observable in the Ising Model and the Inverse
  Kac--Ward Operator}},
        date={2014Oct},
        ISSN={1424-0661},
     journal={Annales Henri Poincar{\'e}},
      volume={15},
      number={10},
       pages={1945\ndash 1965},
         url={https://doi.org/10.1007/s00023-013-0295-z},
}

\bib{Mercat}{article}{
      author={Mercat, C.},
       title={{Discrete Riemann Surfaces and the Ising Model}},
        date={2001Apr},
        ISSN={1432-0916},
     journal={Communications in Mathematical Physics},
      volume={218},
      number={1},
       pages={177\ndash 216},
         url={https://doi.org/10.1007/s002200000348},
}

\bib{Onsager}{article}{
      author={Onsager, Lars},
       title={Crystal statistics. {I}. {A} two-dimensional model with an
  order-disorder transition},
        date={1944},
     journal={Phys. Rev. (2)},
      volume={65},
       pages={117\ndash 149},
      review={\MR{0010315 (5,280d)}},
}

\bib{Schramm}{article}{
      author={Schramm, O.},
       title={Circle patterns with the combinatorics of the square grid},
        date={1997},
        ISSN={0012-7094},
     journal={Duke Math. J.},
      volume={86},
      number={2},
       pages={347\ndash 389},
         url={https://doi.org/10.1215/S0012-7094-97-08611-7},
      review={\MR{1430437}},
}

\bib{Sherman1}{article}{
      author={Sherman, S.},
       title={Combinatorial aspects of the {I}sing model for ferromagnetism.
  {I}. {A} conjecture of {F}eynman on paths and graphs},
        date={1960},
        ISSN={0022-2488},
     journal={J. Mathematical Phys.},
      volume={1},
       pages={202\ndash 217},
      review={\MR{0119512 (22 \#10273)}},
}

\bib{smirnov}{article}{
      author={Smirnov, Stanislav},
       title={Conformal invariance in random cluster models. {I}. {H}olomorphic
  fermions in the {I}sing model},
        date={2010},
        ISSN={0003-486X},
     journal={Ann. of Math. (2)},
      volume={172},
      number={2},
       pages={1435\ndash 1467},
         url={http://dx.doi.org/10.4007/annals.2010.172.1441},
      review={\MR{2680496 (2011m:60302)}},
}

\bib{vdW}{article}{
      author={van~der Waerden, B.~L.},
       title={{Die lange Reichweite der regelm{\"a}{\ss}igen Atomanordnung in
  Mischkristallen}},
        date={1941Jul},
        ISSN={0044-3328},
     journal={Zeitschrift f{\"u}r Physik},
      volume={118},
      number={7},
       pages={473\ndash 488},
         url={https://doi.org/10.1007/BF01342928},
}

\end{biblist}
\end{bibdiv}
\end{document}